\documentclass{article}
\usepackage{amssymb} 
\usepackage{hyperref} 
\usepackage{amsmath} 
\usepackage{amsthm} 

\newtheorem{theorem}{Theorem}
\newtheorem{lemma}{Lemma}

\newcommand{\ket}[1]{|#1\rangle} 
\newcommand{\bra}[1]{\langle#1|} 
\newcommand{\ketpsi}{| \psi \rangle} 
\newcommand{\brapsi}{\langle \psi |} 
\newcommand{\tketpsi}{| \tilde{\psi} \rangle} 

\newcommand{\C}{\ensuremath{\mathbb{C}}}
\newcommand{\N}{\ensuremath{\mathbb{N}}}

\newcommand{\R}{\ensuremath{\mathbb{R}}}

\newcommand{\complexi}{\ensuremath{\mathrm{i}}}

\newcommand{\oneoverroottwo}{ \frac{1}{\sqrt{2}} }
\newcommand{\chshnexpression}{\frac{1}{4 {n \choose 2}} \sum_{1 \leq i < j \leq n}  \brapsi \left( A_i \otimes B_{ij} + A_i \otimes B_{ji} + A_j \otimes B_{ij} - A_j \otimes B_{ji} \right) \ketpsi } 
\newcommand{\poweroftwo}{2^{\lfloor n/2 \rfloor}}
\newcommand{\sumij}{\sum_{1 \leq i < j \leq n}}

\newcommand{\chshn}{CHSH$(n)$ }
\newcommand{\abpsi}{\ensuremath{A_i, B_{jk}, \ketpsi} }
\newcommand{\tildeabpsi}{$\tilde{A}_i, \tilde{B}_{jk}, \ket{\tilde{\psi}}$ }
\newcommand{\aobservables}{$A_i, \: i=1, \dots n$ }
\newcommand{\tildeaobservables}{$\tilde{A}_i, \: i=1, \dots n$ }
\newcommand{\bobservables}{$B_{jk}, \: j \neq k \in \{1, \dots n\}$ }

\renewcommand{\S}{\ensuremath{\mathcal{S}}}

\newcommand{\ai}{A_i \otimes I}

\newcommand{\bkl}{I \otimes B_{kl}}
\newcommand{\tai}{\tilde{A}_i \otimes I}
\newcommand{\tbjk}{I \otimes \tilde{B}_{jk}}
\newcommand{\tbkl}{I \otimes \tilde{B}_{kl}}

\newcommand{\strategyspace}{$\C^{d_A} \otimes \C^{d_B}$ }
\newcommand{\canonicalstrategyspace}{$\C^{2^{\lceil n/2 \rceil}} \otimes \C^{2^{\lceil n/2 \rceil}}$ }
\newcommand{\cda}{\C^{d_A}}
\newcommand{\cdb}{\C^{d_B}}
\newcommand{\cdacdb}{\ensuremath{\C^{d_A} \otimes \C^{d_B}}}

\newcommand{\signij}{sign(i, j_1, \dots j_n)}
\newcommand{\signjk}{sign(j_1, \dots j_n, k)}
\newcommand{\signjl}{sign(j_1, \dots j_n, l)}

\newcommand{\aj}{A_1^{j_1} \dots A_n^{j_n}}
\newcommand{\taj}{\tilde{A}_1^{j_1} \dots \tilde{A}_n^{j_n}}

\newcommand{\tajtensori}{\tilde{A}_1^{j_1} \dots \tilde{A}_n^{j_n} \otimes I}

\newcommand{\tajmodk}{\tilde{A}_1^{j_1} \dots \tilde{A}_k^{j_k \oplus 1} \dots  \tilde{A}_n^{j_n}}

\newcommand{\tajmodl}{\tilde{A}_1^{j_1} \dots \tilde{A}_l^{j_l \oplus 1} \dots  \tilde{A}_n^{j_n}}

\newcommand{\ajpsi}{ A_1^{j_1} \dots A_n^{j_n} \otimes I  \ketpsi}
\newcommand{\tildeajpsi}{ \tilde{A}_1^{j_1} \dots \tilde{A}_n^{j_n} \otimes I  \ket{\tilde{\psi}}}
\newcommand{\tildepsiajinv}{\bra{\tilde{\psi}} \left( \tilde{A}_1^{j_1} \dots \tilde{A}_n^{j_n} \otimes I \right)^{\dagger} }
\newcommand{\jinzeroonetothen}{(j_1 \dots j_n) \in \{ 0, 1 \}^n}
\newcommand{\avsumj}{\frac{1}{\sqrt{2^n}} \sum_{(j_1 \dots j_n) \in \{ 0, 1 \}^n }}
\newcommand{\intertwiningoperatorexpression}{ \avsumj \ajpsi \tildepsiajinv }
\newcommand{\ajpsimodi}{ \signij  A_1^{j_1} \dots A_i^{j_i \oplus 1} \dots  A_n^{j_n} \otimes I  \ketpsi }
\newcommand{\ajpsimodk}{ \signjk A_1^{j_1} \dots A_k^{j_k \oplus 1} \dots  A_n^{j_n} \otimes I  \ketpsi }
\newcommand{\ajpsimodl}{ \signjl  A_1^{j_1} \dots A_l^{j_l \oplus 1} \dots  A_n^{j_n} \otimes I  \ketpsi }
\newcommand{\ajbklpsi}{A_1^{j_1} \dots A_n^{j_n} \otimes B_{kl}  \ketpsi}

\begin{document}

\title{The structure of optimal and nearly-optimal quantum strategies for non-local XOR games}
\author{Dimiter Ostrev\thanks{Department of Mathematics, Massachusetts Institute of Technology.}}
\maketitle

\begin{abstract}
We study optimal and nearly-optimal quantum strategies for non-local XOR games. First, we prove the following general result: for every non-local XOR game, there exists a set of relations with the properties: (1) a quantum strategy is optimal for the game if and only if it satisfies the relations, and (2) a quantum strategy is nearly optimal for the game if and only if it approximately satisfies the relations. Next, we focus attention on a specific infinite family of XOR games: the CHSH(n) games. This family generalizes the well-known CHSH game. We describe the general form of CHSH(n) optimal strategies. Then, we adapt the concept of intertwining operator from representation theory and use that to characterize nearly-optimal CHSH(n) strategies. 
\end{abstract}

\section{Introduction}

Non-local XOR games are a framework used to study the correlations that result from measuring two parts of an entangled quantum state using two spatially separated devices, each capable of performing one of several possible measurements. 

When we think of a non-local XOR game, we imagine two people, usually called Alice and Bob, in two spatially separated laboratories, and unable to communicate with each other. Alice and Bob choose a strategy for the game by choosing a particular setup for their respective measurement devices, and a particular entangled quantum state shared between them. Alice and Bob's aim in choosing their strategy is to maximize a given linear functional acting on the space of correlations. The linear functional represents the rules of the particular XOR game Alice and Bob are playing; the higher the value of the linear functional on the correlations produced by Alice and Bob's strategy, the better Alice and Bob are doing. 

It has long been known that for certain non-local XOR games, Alice and Bob can achieve a higher value using measurements of a shared entangled state than anything Alice and Bob could do using only a classical shared random string (see, for example, the surveys \cite{werner2001bell, brunner2014bell}). This has attracted interest both from the point of view of foundations of physics, and from the point of view of applications. From the point of view of foundations of physics, the advantage of quantum strategies over classical ones has been central in the discussion about local realism (see, for example, the survey \cite{clauser1978bell}). From the point of view of applications, there have been many proposals for using quantum entanglement as a resource in information processing tasks, such as performing distributed computation with a lower communication cost (see, for example, the survey \cite{buhrman2010nonlocality}), teleportation of quantum states \cite{bennett1993teleporting} and the extension to a full scale computation by teleportation scheme \cite{gottesman1999demonstrating}, and quantum cryptography (see, for example, the survey \cite{gisin2002quantum}). 

In the study of non-local XOR games, the optimal and nearly-optimal quantum strategies are interesting objects for several reasons. First, their behavior is maximally far away from the behavior of classical strategies. Second, applications often involve setups related to the optimal strategies. Third, the optimal quantum strategies represent the boundary of the non-local correlations that are achievable in quantum mechanics, and are therefore interesting from the perspective of foundations of quantum mechanics. And finally, the optimal and nearly optimal quantum strategies for XOR games have interesting mathematical structure, with connections to semi-definite programming and representation theory.  

In this paper, we study the optimal and nearly optimal quantum strategies for non-local XOR games. First, we present the following general result: for every non-local XOR game, there exists a set of relations such that 
\begin{enumerate}
\item A strategy is optimal for the game if and only if it satisfies the relations. 
\item A strategy is nearly-optimal for the game if and only if it approximately satisfies the relations. 
\end{enumerate}
The coefficients of the relations can be computed efficiently by solving a semi-definite program and finding the eigenvalues and eigenvectors of a positive semi-definite matrix. The precise statement is in Theorem \ref{thm:equality_conditions} and the proof in Section \ref{ch:relations_for_strategies}. 

The result in Theorem \ref{thm:equality_conditions} continues the line of work in references \cite{tsirel1987quantum,cleve2004consequences,slofstra2011lower}. In \cite{tsirel1987quantum}, a correspondence was established between the quantum non-local correlations and inner products of vectors in real euclidean space. Later, in \cite{cleve2004consequences}, it was noticed that a semi-definite program can be associated to each non-local XOR game. In reference \cite{slofstra2011lower}, the dual semi-definite program was used to obtain the so-called marginal biases for a non-local XOR game. In this paper, we use the dual semi-definite program to derive the set of relations for optimal and near-optimal quantum strategies of a given XOR game.

In the second part of this paper, we focus on a specific infinite family of non-local XOR games: the \chshn games, $n \in \N, \, n \geq 2$ introduced in \cite{slofstra2011lower}. For this family, we solve the system of relations mentioned above, and precisely characterize the optimal and nearly-optimal \chshn strategies. 

The interest in precisely characterizing optimal and nearly-optimal quantum strategies for XOR games comes from recent results about information processing with untrusted black-box quantum devices. In these results, one or more parties attempt to perform an information processing task, such as quantum key distribution, randomness generation, or distributed computation, by interacting via classical inputs and outputs with quantum devices that cannot be trusted to perform according to specification. The devices may not be trusted for example for fear of malicious intent, as in quantum cryptography, or, to take another example, the manufacturing process used to make the devices may be unreliable and prone to errors. 

The task of doing information processing with untrusted black-box devices and being confident in the result may at first appear daunting. However, there have recently been proposals of protocols for quantum key distribution with untrusted devices, for randomness generation with untrusted devices, and for a protocol in which a classical verifier commands two untrusted quantum provers to perform a full-scale quantum computation. References to results of this type may be found for example as follows: for quantum key distribution, the original proposals are \cite{mayers1998quantum,mayers2003self}, a more recent result is \cite{vazirani1810fully}, and the survey \cite{brunner2014bell} lists a number of other results on p.34-35; for randomness generation, the survey \cite{brunner2014bell} lists a number of results on p.33; the protocol in which a classical verifier commands two untrusted quantum provers to perform a full-scale quantum computation is developed in reference \cite{reichardt2012classical}.

All of these protocols rely on mathematical results that have been given the name of self-testing or entanglement rigidity (see \cite{reichardt2012classical,mckague2012robust,miller2012optimal} for three examples of such results, with different proof techniques in each). These results are a characterization of optimal and nearly-optimal strategies for the CHSH game (or close cousins of the CHSH game). The CHSH game is the first member of the family \chshn, $n \geq 2$, mentioned above. 

In this paper we obtain a precise characterization of optimal and nearly-optimal strategies for all the \chshn XOR games. The techniques used in the proof differ from the self-testing results mentioned above; here we use ideas form representation theory. 

It has been noticed previously \cite{tsirel1987quantum, slofstra2011lower} that representation theory is well-suited to describing exactly optimal quantum strategies for non-local XOR games. In the case of exactly optimal \chshn strategies, the contribution of this paper is to give an explicit and direct statement and proof of a classification theorem for the \chshn exactly optimal strategies. The precise statement is in Theorem \ref{thm:chshn_optimal_strategies}, and the proof in Section \ref{sec:chshn_optimal_strategies}. 

The situation with nearly-optimal strategies is more subtle; the representation theory techniques that work so well in the exact case are difficult to generalize to nearly-optimal strategies (we will say more about the difficulty later). An attempt to use representation theory in this context has been made in \cite{slofstra2011lower}, but the error bounds obtained there depend on the dimension of the Hilbert space used for the strategy; in the context of untrusted black box devices, this dimension may be arbitrarily large. 

In this paper, we take a different approach to characterizing nearly-optimal quantum strategies. The key insights are to adapt the concept of intertwining operator from representation theory, to notice the importance of a certain subspace of the space of a given strategy and to adapt the group averaging technique from representation theory. The precise statement of the result for \chshn near-optimal strategies is in Theorem \ref{thm:chshn_near_optimal_strategies}, and the proof in Section \ref{sec:chshn_near_optimal_strategies}. 

The remainder of this paper is structured as follows: in Section \ref{ch:preliminaries}, we present notation, concepts and known facts that are necessary background for the rest of the paper. In Section \ref{ch:overview_of_results}, we give the precise statements of the results proved in this paper. Sections \ref{ch:relations_for_strategies}, \ref{sec:chshn_optimal_strategies}, \ref{sec:chshn_near_optimal_strategies} contain the proofs of the main results. In Section \ref{ch:open_problems} we discuss open problems and possible future work. 

\section{Preliminaries} \label{ch:preliminaries}

The goal of this section is to cover notation, concepts and known facts that are used throughout the rest of the paper. 

\subsection{A linear bijection between $\C^{d_A} \otimes \C^{d_B}$ and $Mat_{d_A, d_B} (\C)$}\label{subsec:a_linear_bijection}

We consider the space $\C^{d_A}$ with its standard basis denoted by $\ket{i}, \: i=1, \dots d_A$ and the space $\C^{d_B}$ with its standard basis denoted by $\ket{j}, \: j=1, \dots d_B$. 

With this notation, we can write the standard basis of $\C^{d_A} \otimes \C^{d_B}$ as \[\ket{i} \otimes \ket{j}, \:  i=1, \dots d_A, \: j=1, \dots d_B\] and we can write the standard basis of $Mat_{d_A, d_B} (\C)$ as \[\ket{i} \bra{j}, \:  i=1, \dots d_A, \: j=1, \dots d_B\]

We define a linear bijection \[ \mathcal{L}\: : \:  \C^{d_A} \otimes \C^{d_B} \longrightarrow Mat_{d_A, d_B} (\C) \] by defining the action of $\mathcal{L}$ on the standard basis as \[ \mathcal{L} \left( \ket{i} \otimes \ket{j} \right) = \ket{i} \bra{j} \] and extending to the whole space by linearity; that is, \[ \mathcal{L} \left( \sum_{ij} w_{ij} \ket{i} \otimes \ket{j} \right) = \sum_{ij} w_{ij} \ket{i} \bra{j} \]

We collect some useful properties of $\mathcal{L}$ in the following lemma. 

\begin{lemma}\label{lemma:properties_of_linear_bijection}
Let $\ket{u} \in \C^{d_A}$, $\ket{v} \in \C^{d_B}$, $\ket{w} \in \C^{d_A} \otimes \C^{d_B}$, $A \in Mat_{d_A} (\C)$, $B \in Mat_{d_B} (\C)$. Then, 
\begin{itemize}
\item $ \mathcal{L}(\ket{u} \otimes \ket{v}) = \ket{u} \bra{v^*} $ and consequently, by linearity, $ \mathcal{L}\left( \sum_{l=1}^k \ket{u_l} \otimes \ket{v_l} \right) = \sum_{l=1}^k \ket{u_l} \bra{v_l^*} $
\item $ A \mathcal{L}(\ket{w}) = \mathcal{L} ( A \otimes I \ket{w} ) $
\item $ \mathcal{L}(\ket{w}) B^T = \mathcal{L}( I \otimes B \ket{w} ) $
\item $ \| \mathcal{L}(\ket{w}) \|_F = \| \ket{w} \| $
\end{itemize}
\end{lemma}

All of these properties can be proved by expanding the relevant vectors and matrices with respect to the standard basis and checking that the appropriate identity in the coefficients holds. 

The notation $\| \; \|_F$ used above denotes the Frobenius norm of a matrix: for an $m \times n$ matrix $A$, 
\[ \|A\|_F = \sqrt{\sum_{i=1}^m \sum_{j=1}^n |a_{ij}|^2} = \sqrt{Tr A^\dagger A} \]

\subsection{Non-local XOR games and their quantum strategies}

In a non-local XOR game two players, traditionally called Alice and Bob, are separated in space and play cooperatively without communicating with each other. A third party, called a Referee or sometimes a Verifier, runs the game and decides whether Alice and Bob win or lose. 

Formally, a non-local game consists of two finite sets $S$ and $T$, a probability distribution $\pi$ on $S \times T$, and a function $V \: : \: S \times T \rightarrow \{-1, 1\}$. The game proceeds as follows: \begin{enumerate}
\item The referee selects a pair $(s,t) \in S \times T$ according to the probability distribution $\pi$. 
\item The referee sends $s$ as a question to Alice and $t$ as a question to Bob. 
\item Alice replies to the referee with $a \in \{-1, 1\}$ and Bob replies to the referee with $b \in \{-1, 1\}$
\item The referee looks at $V(s,t)ab$. If $V(s,t)ab=1$, then Alice and Bob win, and if $V(s,t)ab=-1$ then Alice and Bob lose. Notice that $V(s,t)=1$ means that Alice and Bob must give matching answers to win and $V(s,t)=-1$ means Alice and Bob must give opposite answers to win.\footnote{The name "XOR game" is related to the following: if we write $a=(-1)^{a'}, \, b=(-1)^{b'}$ for $a' , b' \in \{0,1\}$, then $V(s,t)ab=V(s,t) (-1)^{a' \oplus b'}$ so that whether Alice and Bob win or lose depends on the XOR of the bits $a'$ and $b'$} 
\end{enumerate}

It is convenient to summarize all the information for an XOR game into a $|S| \times |T|$ matrix $G$ such that $G_{st} = \pi (s,t) V(s,t)$. The matrix $G$ contains all the information about the game: the set $S$ is the set of row indices of $G$, the set $T$ is the set of column indices of $G$, the probability distribution $\pi$ can be recovered by $\pi(s,t) = |G_{st}|$, the function $V$ can be recovered by $V(s,t) = sign ( G_{st})$. Thus, we can identify non-local XOR games with matrices $G$ normalized so that $\sum_{st} |G_{st}| = 1$. 

A quantum strategy $\mathcal{S}$ for an XOR game consists of a state space $\cdacdb$, a state $\ketpsi \in \cdacdb$, and $\pm 1$ observables $\{A_s \, : \, s \in S \}$ on $\cda$ and $\{B_t \, : \, t \in T \}$ on $\cdb$. The interpretation of this strategy is the following: Alice and Bob share a bipartite quantum system with state space $\cdacdb$. Prior to the beginning of the game, the system has been prepared in the state $\ketpsi \in \cdacdb$. On receiving question $s$, Alice measures observable $A_s$ and uses the outcome, $1$ or $-1$, as her answer to the referee. Similarly, on receiving question $t$, Bob measures observable $B_t$ and uses the outcome, $1$ or $-1$, as his answer to the referee.

We would like to have a way to evaluate how well a given strategy $\mathcal{S}$ does for a given XOR game $G$. We do so using the success bias $\beta(G, \S)$ defined by: \[\beta(G, \mathcal{S}) = \sum_{s \in S} \sum_{t \in T} G_{st} \brapsi A_s \otimes B_t \ketpsi \] The success bias is linearly related to the probability $\omega(G, \S)$ of winning $G$ using strategy \S: \[ \beta(G, \S) = 2 \omega(G, \S) -1\]

We define the quantum success bias $\beta(G)$ for an XOR game $G$ to be the supremum of the success bias over all quantum strategies: \[\beta(G) = \sup_{\S} \beta(G, \S) \] 

We define an optimal strategy for the XOR game $G$ to be a strategy $\S$ such that \[ \beta(G, \S) = \beta(G) \] and we define an $\epsilon$-optimal strategy to be a strategy $\S$ such that \[ (1-\epsilon) \beta(G) \leq \beta(G, \S) \leq \beta(G) \]

\subsection{The \chshn XOR games}\label{subsec:definition_of_chshn}

Here, we look at the infinite family of XOR games \chshn, $n \in \N, \, n \geq 2$ introduced in \cite{slofstra2011lower}. 

For the \chshn game, the set $S$ of possible questions for Alice is $\{1, \dots, n\}$ and the set $T$ of possible questions for Bob is the set of ordered pairs $\{ij \: : \: i,j \in \{1, \dots, n\}, \: i \neq j\}$. 

The referee selects questions according to the following probability distribution $\pi(s,t)$: 
\begin{enumerate}
\item The referee selects a pair $i,j$ uniformly at random among all ${n \choose 2}$ pairs such that $1 \leq i < j \leq n$. 
\item The referee selects either $i$ or $j$ as question for Alice, and either $ij$ or $ji$ as question for Bob; the four possibilities are equally likely. 
\end{enumerate}

The rule for winning or losing $V(s,t)$ is determined like this: to win, Alice and Bob must give matching answers on questions $(i, ij)$, $(i, ji)$ and $(j, ij)$, and give opposite answers on questions $(j, ji)$. 

As in the previous subsection, it is convenient to summarize all information about the \chshn game in a matrix $G$. The matrix $G$ for the \chshn game has $n$ rows and $n(n-1)$ columns. It is most convenient to write the matrix $G$ using Dirac's bra-ket notation. Let $\ket{1}, \dots \ket{n}$ be an orthonormal basis of $\R^n$, and let $\ket{ij}, \: i \neq j \in \{1, \dots n\}$ be an orthonormal basis of $\R^{n(n-1)}$. Then, we can write:
\[ G = \frac{1}{4{n \choose 2}} \sum_{1 \leq i < j \leq n} \Big( \ket{i} \bra{ij} + \ket{j} \bra{ij} + \ket{i} \bra{ji} - \ket{j} \bra{ji} \Big)\]

It was shown in reference \cite{slofstra2011lower} that the quantum success bias for all the \chshn games is $\oneoverroottwo$; that is, \[ \sup_{\abpsi} \chshnexpression = \oneoverroottwo \]

Finally, we note that the first element of the family, CHSH(2), is the usual CHSH game, based on reference \cite{clauser1969proposed}. Thus, the family \chshn is a generalization of the CHSH game.

\subsection{Semi-definite programs}

In this section we cover some terminology and facts about semi-definite programs that will be used later on. We use an abbreviated discussion on semi-definite programs that is sufficient for the purposes of this paper; for a more detailed exposition see, for example, \cite{vandenberghe1996semidefinite}, or the lecture notes \cite{lovazsSDPnotes}. 

Look at the space of real symmetric matrices of a given size. For two such matrices $A$, $B$, we define their inner product \[A \cdot B = Tr \: AB = \sum_{ij} A_{ij} B_{ij} \]

Within the space of real symmetric matrices, we look at the positive semi-definite matrices. We use the notation  $A \succeq 0$ to mean that $A$ is positive semi-definite, and the notation $A \succ 0$ to mean that $A$ is strictly positive definite. This notation also extends in the following way: $A \succeq B$ means that $(A-B)$ is positive semi-definite and  $A \succ B$ means that $(A-B)$ is strictly positive definite. 

A semi-definite program is a constraint optimization problem of the form \[ \sup_{Z \succeq 0, \: F_i \cdot Z = c_i, \, i=1, \dots m} G \cdot Z \] Here $G$, $F_i, \, i=1, \dots m$ are symmetric matrices, and $c_i, \, i=1, \dots m$ are real numbers. We call this semi-definite program the primal. We denote the value of the supremum by $v_{primal}$. 

The dual semi-definite program is \[ \inf_{\sum_{i=1}^m y_i F_i \succeq G} \vec{c} \cdot \vec{y}\] We denote the value of the infimum by $v_{dual}$. 

Next, we introduce some terminology: 
\begin{itemize}
\item A primal/dual feasible solution is one that satisfies the constraints.
\item A primal/dual strictly feasible solution is one that satisfies the constraints, and satisfies the positive semi-definite constraint strictly. 
\item A primal/dual optimal solution is a feasible solution $Z$, respectively $\vec{y}$, such that $G \cdot Z = v_{primal}$, respectively $\vec{c} \cdot \vec{y} = v_{dual}$
\item A primal/dual $\epsilon$-optimal solution is a feasible solution $Z$, respectively $\vec{y}$, such that $G \cdot Z \geq (1-\epsilon) v_{primal}$, respectively $\vec{c} \cdot \vec{y} \leq (1+\epsilon) v_{dual}$.
\item For a primal feasible $Z$ and a dual feasible $\vec{y}$, the quantity \[ \left( \sum_{i=1}^m y_i F_i - G \right) \cdot Z\] is called the duality gap. 
\end{itemize}

We summarize some known facts about semi-definite programs in the following theorem:

\begin{theorem}\label{thm:properties_of_sdp}
Assume throughout that both the primal and the dual have feasible solutions. The following statements hold
\begin{itemize}
\item For a primal feasible $Z$ and a dual feasible $\vec{y}$, the duality gap is non-negative: \[ \left( \sum_{i=1}^m y_i F_i - G \right) \cdot Z \geq 0\] 
\item $\left( \sum_{i=1}^m y_i F_i - G \right) \cdot Z = 0$ if and only if $v_{primal}=v_{dual}$, $Z$ is optimal for the primal and $\vec{y}$ optimal for the dual. This statement is sometimes called "complementary slackness condition". 
\item $v_{primal} \leq v_{dual}$. This statement is sometimes called "weak duality". 
\item If the primal has a strictly feasible solution, then the dual infimum is attained; if the dual has a strictly feasible solution, then the primal supremum is attained.  
\item If at least one of the primal and dual has a strictly feasible solution, then $v_{primal} = v_{dual}$. This statement is sometimes called "strong duality". 
\end{itemize}
\end{theorem}

\subsection{Some facts from representation theory}

In this section we cover a few facts and concepts from representation theory that will be used later on. These facts include properties of anti-commuting $\pm 1$ observables, invariant subspaces and Schur's lemma, and the notion of an intertwining operator. For a more detailed exposition of representation theory, see for example \cite{fulton1991representation} or the lecture notes \cite{Etingof2011RepTheoryLectureNotes}. 

\subsubsection{$2k+1$ anti-commuting $\pm 1$ observables on $\C^{2^k}$}

We give an explicit construction of $2k+1$ anti-commuting $\pm 1$ observables on $\C^{2^k}$ using the isomorphism $\C^{2^k} \cong \underbrace{\C^2 \otimes \C^2 \otimes \dots \otimes \C^2}_{k \, terms}$ and the Pauli matrices. 

Consider the following $2 k +1$ operators on $\C^{2^k} \cong \underbrace{\C^2 \otimes \C^2 \otimes \dots \otimes \C^2}_{k \, terms}$:
\begin{equation}\label{eq:sigma_anticommuting_observables}
\begin{aligned}
 \sigma_{k,1} &= \sigma_x \otimes I \otimes I \otimes I \otimes \dots \otimes I \otimes I \\
 \sigma_{k,2} &= \sigma_z \otimes I \otimes I \otimes I \otimes \dots \otimes I \otimes I \\
 \sigma_{k,3} &= \sigma_y \otimes \sigma_x \otimes I \otimes I \otimes \dots \otimes I \otimes I \\
 \sigma_{k,4} &= \sigma_y \otimes \sigma_z \otimes I \otimes I \otimes \dots \otimes I \otimes I \\
 \sigma_{k,5} &= \sigma_y \otimes \sigma_y \otimes \sigma_x \otimes I \otimes \dots \otimes I \otimes I \\
&\cdots \\
\sigma_{k,2k-1} &= \sigma_y \otimes \sigma_y \otimes \sigma_y \otimes \sigma_y \otimes \dots \otimes \sigma_y \otimes \sigma_x \\
\sigma_{k,2k} &= \sigma_y \otimes \sigma_y \otimes \sigma_y \otimes \sigma_y \otimes \dots \otimes \sigma_y \otimes \sigma_z \\
\sigma_{k,2k+1} &= \sigma_y \otimes \sigma_y \otimes \sigma_y \otimes \sigma_y \otimes \dots \otimes \sigma_y \otimes \sigma_y 
\end{aligned}
\end{equation}
These operators are self-adjoint, unitary, and anti-commute. 

It is known from the representation theory of the Clifford algebra that any collection of $2k$ anti-commuting $\pm1$ observables on $\C^{2^k}$ is equivalent (by conjugation by unitary) to the collection $\sigma_{k,1}, \dots \sigma_{k,2k}$, and any collection of $2k+1$ anti-commuting $\pm1$ observables on $\C^{2^k}$ is equivalent to either $\sigma_{k,1}, \dots \sigma_{k,2k}, \sigma_{k,2k+1}$ or $\sigma_{k,1}, \dots \sigma_{k,2k}, -\sigma_{k,2k+1}$ (the two options are not equivalent because the product of the observables in the first collection is $(-\complexi)^k I$ and the product in the second collection is $-(-\complexi)^k I$). 

\subsubsection{The general form of $n$ anti-commuting $\pm1$ observables on $\C^d$}

It follows from the representation theory of the Clifford algebra that the following holds for $n$ anti-commuting $\pm1$ observables on $\C^d$:

\begin{theorem}\label{thm:the_general_form_of_n_anti_commuting_observables}
Let $A_1, \dots A_n$ be $\pm1$ observables on $\C^d$ such that $A_k A_l + A_l A_k = 0$ for $k \neq l$. Then $d = s \poweroftwo$ for some $s \in \N$, and there is an orthonormal basis of $\C^d$ with respect to which $A_1, \dots A_n$ have block-diagonal form with $\poweroftwo \times \poweroftwo$ blocks and such that 
\begin{itemize}
\item For $n=2 k$, $i=1, \dots 2k$, the diagonal blocks of $A_i$ are all equal to $\sigma_{k,i}$. 
\item For $n=2 k + 1$, $i=1, \dots 2k$, the diagonal blocks of $A_i$ are all equal to $\sigma_{k,i}$, and for $i=2k+1$ some number $s' , \, 0 \leq s' \leq s$ of the diagonal blocks of $A_{2k+1}$ are $ \sigma_{k,2k+1} $ and the other $s-s'$ diagonal blocks are  $ - \sigma_{k,2k+1} $
\end{itemize}
\end{theorem}

\subsubsection{Anti-commuting $\pm1$ observables and inner products}\label{subsec:anti_commuting_observables_and_inner_products}

Here we present a property relating $n$ anti-commuting $\pm1$ observables and inner products of vectors in $\R^n$. We introduce a piece of notation and then state the property. 

Let $A_1, \dots A_n$ be some matrices on $\C^d$, and let $u=\begin{bmatrix}
u_1 & \dots & u_n
\end{bmatrix}^T$ be a vector in $\R^n$. By $u \cdot \vec{A}$ we mean a linear combination of $A_1, \dots A_n$ with the coefficients $u_1, \dots u_n$; that is, \[u \cdot \vec{A} = u_1 A_1 + \dots + u_n A_n\] 

With this notation, we can state the following lemma:

\begin{lemma}\label{lemma:anti_commuting_observables_and_inner_products}
Let $A_1, \dots A_n$ be anti-commuting $\pm1$ observables on $\C^d$, let $\ketpsi \in \C^d \otimes \C^d$ be the maximally entangled state $\ketpsi = \frac{1}{\sqrt{d}} \sum_{i=1}^{d} \ket{ii}$ and let $u, v\in \R^n$ be two vectors. Then 
\begin{enumerate}
\item $(u \cdot \vec{A})(v \cdot \vec{A}) + (v \cdot \vec{A})(u \cdot \vec{A}) = 2 \left( \sum_{i=1}^n u_i v_i \right) I = 2 (u^T v) I$
\item $ \brapsi (u \cdot \vec{A}) \otimes (v \cdot \vec{A})^T \ketpsi = u^T v $
\end{enumerate}
\end{lemma}
\begin{proof}
For part 1: expand the left-hand-side: \[  (u \cdot \vec{A})(v \cdot \vec{A}) + (v \cdot \vec{A})(u \cdot \vec{A}) =  \sum_{i=1}^n \sum_{j=1}^n u_i v_j (A_i A_j + A_j A_i) = 2 \left( \sum_{i=1}^n u_i v_i \right) I\]

For part 2: the maximally entangled state $\ketpsi$ has the property $M \otimes I \ketpsi = I \otimes M^T \ketpsi$ (and consequently also $\brapsi M \otimes I = \brapsi I \otimes M^T$) for any matrix $M$ on $\C^d$. Then, 
\begin{multline*}
\brapsi (u \cdot \vec{A}) \otimes (v \cdot \vec{A})^T \ketpsi = \brapsi \frac{(u \cdot \vec{A})(v \cdot \vec{A}) + (v \cdot \vec{A})(u \cdot \vec{A})}{2} \otimes I \ketpsi \\ 
= \brapsi (u^T v) I \otimes I \ketpsi = u^T v 
\end{multline*}
\end{proof}

\subsubsection{Invariant subspaces and Schur's lemma}

Here we present some facts about invariant subspaces. These facts are commonly called Schur's lemma in expositions of representation theory. We introduce the notion of invariant subspace and then state Schur's lemma. 

Let $A$ be a matrix on $\C^d$ and let $V$ be a subspace of $\C^d$. We say that $V$ is invariant under $A$ if \[ \ket{v} \in V \; \Rightarrow \; (A \ket{v} ) \in V\] This also generalizes to a collection of matrices: let $\mathcal{I}$ be some index set and let $\{A_i \, : \, i \in \mathcal{I} \}$
be a collection of matrices. We say that $V$ is invariant under the collection $\{A_i \, : \, i \in \mathcal{I} \}$ if it is invariant under each individual $A_i$. In the context of representation theory, the index set $\mathcal{I}$ has the extra structure of being a group or an algebra, and the mapping $i \mapsto A_i$ has the extra structure of being a group or algebra homomorphism. However, this extra structure is not used in the proof of Schur's lemma, and the lemma holds for general index sets $\mathcal{I}$. 

Now we are ready to state Schur's lemma: 

\begin{lemma}\label{lemma:Schurs_lemma}
\begin{enumerate}
\item Let  $\{A_i \, : \, i \in \mathcal{I} \}$ be a collection of linear operators on $V$, let $\{B_i \, : \, i \in \mathcal{I} \}$ be a collection of linear operators on $W$, and let $T$ be a linear operator $V \rightarrow W$. Suppose $T A_i = B_i T$ for all $i \in \mathcal{I}$. Then $Im T$ is invariant under the collection $\{B_i \, : \, i \in \mathcal{I} \}$ and $Ker T$ is invariant under the collection $\{A_i \, : \, i \in \mathcal{I} \}$. 
\item Let $\{A_i \, : \, i \in \mathcal{I} \}$ be a collection of linear operators on $V$ and $T$ be a linear operator on $V$. Suppose $A_i T = T A_i$ for all $i \in \mathcal{I}$. Then all eigenspaces of $T$ are invariant under the collection $\{A_i \, : \, i \in \mathcal{I} \}$. 
\end{enumerate}
\end{lemma}

These statements can be proved directly from the definitions. 

\subsubsection{Intertwining operators}

Here we look at the concept of intertwining operator that is implicitly present in the statement of Schur's lemma. 

Let $\{A_i \, : \, i \in \mathcal{I} \}$ be a collection of linear operators on $V$, $\{B_i \, : \, i \in \mathcal{I} \}$ be a collection of linear operators on $W$, and $T$ a linear operator $V \rightarrow W$. We say that $T$ is an intertwining operator for the collections $\{A_i \, : \, i \in \mathcal{I} \}$, $\{B_i \, : \, i \in \mathcal{I} \}$ if $T A_i = B_i T$ for all $i \in \mathcal{I}$. 

In the context of representation theory, the index set $\mathcal{I}$ has the extra structure of being a group or an algebra, and the mappings $i \mapsto A_i$, $i \mapsto B_i$ have the extra structure of being group or algebra homomorphisms. Here, we will want the slightly more general definition that allows an arbitrary index set $\mathcal{I}$. 

\section{Overview of Results}\label{ch:overview_of_results}

\subsection{Relations for strategies}

First, we look at the question: given a non-local XOR game, what can we say about optimal and nearly optimal strategies for the game? We prove the following: 

\begin{theorem} \label{thm:equality_conditions}
Consider a non-local XOR game specified by an $n \times m$ matrix $G$ and with quantum success bias $\beta(G)$. Then, there exist vectors $u_1, \dots u_r \in \R^n$ and $v_1, \dots v_r \in \R^m$ with the property: $\pm 1$ observables $A_1, \dots A_n$, $B_1, \dots B_m$ and bipartite state $\ketpsi$ are an $\epsilon$-optimal strategy for the game, i.e, 
\[ (1-\epsilon) \beta(G)  \leq \sum_{i=1}^{n} \sum_{j=1}^{m} G_{ij} \brapsi A_i \otimes B_j \ketpsi \leq \beta(G) \] if and only if 
\[ \sum_{k=1}^{r} \left\| u_k \cdot \vec{A} \otimes I \ketpsi - I \otimes v_k \cdot \vec{B} \ketpsi \right\|^2 \leq \beta(G) \epsilon \] 
By taking $\epsilon=0$, it follows that a strategy is optimal if and only if \[\forall k=1, \dots r \quad u_k \cdot \vec{A} \otimes I \ketpsi = I \otimes v_k \cdot \vec{B} \ketpsi\]
\end{theorem} 

The proof of Theorem \ref{thm:equality_conditions} is in Section \ref{ch:relations_for_strategies}. The proof relies on the semi-definite program that can be associated to an XOR game, and on an argument that is related to the complementary slackness condition. From the proof, one can see that the vectors $u_1, \dots u_r \in \R^n$ and $v_1, \dots v_r \in \R^m$ from the statement of Theorem \ref{thm:equality_conditions} can be computed efficiently by solving a semi-definite program and finding the eigenvalues and eigenvectors of a positive semi-definite matrix. 

Next, we focus attention on the CHSH($n$) XOR games. By specializing the methods form the proof of Theorem \ref{thm:equality_conditions} to the case of CHSH$(n)$, we obtain the following theorem:

\begin{theorem} \label{thm:chshn_equality_conditions}
The following three statements for $\pm 1$ observables $A_i, B_{jk}$ and bipartite state $\ketpsi$ are equivalent:
\begin{itemize}
\item $A_i, B_{jk}, \ketpsi$ is an $\epsilon$-optimal CHSH($n$) strategy, i.e. 
\begin{multline*} 
\frac{1}{\sqrt{2}} (1-\epsilon) \\ \leq \chshnexpression \leq \frac{1}{\sqrt{2}}
\end{multline*}
\item The observables and state satisfy
\begin{multline*} 
\sum_{1 \leq i < j \leq n} \Bigg( \left\| \frac{A_i + A_j}{\sqrt{2}} \otimes I \ketpsi - I \otimes B_{ij} \ketpsi \right\|^2  \\ 
+ \left\| \frac{A_i - A_j}{\sqrt{2}} \otimes I \ketpsi - I \otimes B_{ji} \ketpsi \right\|^2 \Bigg) \leq 2 n (n-1) \epsilon 
\end{multline*}
\item The observables and state satisfy
\begin{multline*}
\sum_{1 \leq i < j \leq n} \Bigg( \left\| A_i \otimes I \ketpsi - I \otimes \frac{B_{ij} + B_{ji}}{\sqrt{2}} \ketpsi \right\|^2  \\
+ \left\| A_j \otimes I \ketpsi - I \otimes \frac{B_{ij} - B_{ji}}{\sqrt{2}} \ketpsi \right\|^2 \Bigg) \leq 2 n (n-1) \epsilon 
\end{multline*}
\end{itemize}
\end{theorem}

Again, taking $\epsilon=0$ we can obtain the relations for exactly optimal \chshn strategies. The proof of Theorem \ref{thm:chshn_equality_conditions} is in Section \ref{sec:chshn_equality_conditions}. 

\subsection{Classification of \chshn optimal strategies}

For the case of \chshn optimal strategies, we obtain the following classification theorem:

\begin{theorem} \label{thm:chshn_optimal_strategies}
\abpsi is an optimal \chshn strategy on the space $\C^{d_A} \otimes \C^{d_B}$ if and only if there exist an orthonormal basis $\ket{u_1}, \dots \ket{u_{d_A}}$ of $\cda$ and an orthonormal basis $\ket{v_1}, \dots \ket{v_{d_B}}$ of $\cdb$ such that all of the following statements hold
\begin{itemize}
\item The non-zero terms in the Schmidt decomposition of $\ketpsi$ are
\[ \sum_{i=1}^{s 2^{\lfloor n/2 \rfloor}} \sqrt{\lambda_i} \ket{u_i} \otimes \ket{v_i}  \]
with the Schmidt coefficients equal in blocks of length $2^{\lfloor n/2 \rfloor}$, i.e. 
\begin{align*}
\lambda_1 = &\dots = \lambda_{\poweroftwo} \\
\lambda_{\poweroftwo +1} = &\dots = \lambda_{2 \cdot \poweroftwo} \\
&\dots \\
\lambda_{(s-1) \poweroftwo +1} = &\dots = \lambda_{s \poweroftwo}
\end{align*}
\item With respect to the basis $\ket{u_1}, \dots \ket{u_{d_A}}$ of $\cda$, the observables \aobservables have the block diagonal form:
\[ A_i = \begin{bmatrix}
A_i^{(1)} &&& \\
& \ddots && \\
&& A_i^{(s)} \\
&&& C_i
\end{bmatrix} \]
where each $A_i^{(j)}$ is $\poweroftwo \times \poweroftwo$ and acts on $span(\ket{u_{(j-1) \poweroftwo +1}}, \dots \ket{u_{j \poweroftwo}})$, and, for each $i = 1, \dots n$, for each $j=1, \dots s$, $A_i^{(j)} = \sigma_{\lfloor n/2 \rfloor, i}$ \footnote{Here, the observables $\sigma_{k,i}$ are the ones defined in the relations \eqref{eq:sigma_anticommuting_observables}. } except for the case $n =2k+1$, and $i=2k+1$, in which case the blocks $A_i^{(j)}$ are either $\sigma_{k,2k+1}$ or $-\sigma_{k,2k+1}$. The block $C_i$ is an arbitrary $\pm 1$ observable on the orthogonal complement of $span(\ket{u_{1}}, \dots \ket{u_{s \poweroftwo}})$. 
\item With respect to the basis $\ket{v_1}, \dots \ket{v_{d_B}}$ of $\cdb$, the observables \bobservables have the block diagonal form:
\[ B_{jk} = \begin{bmatrix}
B_{jk}^{(1)} &&& \\
& \ddots && \\
&& B_{jk}^{(s)} \\
&&& D_{jk}
\end{bmatrix} \]
where each $B_{jk}^{(l)}$ is $\poweroftwo \times \poweroftwo$ and acts on \\ $span(\ket{v_{(j-1) \poweroftwo +1}}, \dots \ket{v_{j \poweroftwo}})$, and, for $1 \leq j <k \leq n$, 
\[ B_{jk}^{(l)} = \left(\frac{A_j^{(l)} + A_k^{(l)}}{\sqrt{2}}\right)^T \quad  B_{kj}^{(l)} = \left(\frac{A_j^{(l)} - A_k^{(l)}}{\sqrt{2}}\right)^T \]
The block $D_{jk}$ is an arbitrary $\pm 1$ observable on the orthogonal complement of  $span(\ket{v_{1}}, \dots \ket{v_{s \poweroftwo}})$. 
\end{itemize}
\end{theorem}

The proof of Theorem \ref{thm:chshn_optimal_strategies} is in Section \ref{sec:chshn_optimal_strategies}. The proof uses the relations from Theorem \ref{thm:chshn_equality_conditions} and the linear bijection $\mathcal{L}$ between $\C^{d_A} \otimes \C^{d_B}$ and $Mat_{d_A, d_B} (\C)$ from subsection \ref{subsec:a_linear_bijection}. The bipartite state $\ketpsi$ from an optimal \chshn strategy is shown to be such that $\Psi = \mathcal{L}(\ketpsi)$ is an intertwining operator between certain linear combinations of Alice's observables and certain linear combinations of the transpose of Bob's observables. Given the special structure of the relations for the \chshn game, this is enough to imply the conclusions of Theorem \ref{thm:chshn_optimal_strategies}.

One way to interpret Theorem \ref{thm:chshn_optimal_strategies} is that any optimal \chshn strategy must be a direct sum of elementary optimal strategies on $\C^{\poweroftwo} \otimes \C^{\poweroftwo}
$, possibly with some additional dimensions on each side that are orthogonal to the support of the state. Another interpretation is that the space $supp_A \ketpsi \otimes supp_B \ketpsi \subseteq \C^{d_A} \otimes \C^{d_B}$\footnote{$supp_A \ketpsi $ is the span of the $A$-side Schmidt vectors of $\ketpsi$ with non-zero Schmidt coefficients, and $supp_B \ketpsi $ is the span of the $B$-side Schmidt vectors of $\ketpsi$ with non-zero Schmidt coefficients.} is a "good subspace" on which the observables from the strategy are "well-behaved": the \aobservables leave the space $supp_A \ketpsi$ invariant and satisfy the canonical anti-commutation relations on that space, and the \bobservables leave the space $supp_B \ketpsi$ invariant and are determined there by $B_{jk} = (A_j^T \pm A_k^T) / \sqrt{2}$. 

\subsection{\chshn nearly-optimal strategies}

We now turn attention to $\epsilon$-optimal \chshn strategies. One may at first hope that an approximate version of Theorem \ref{thm:chshn_optimal_strategies} holds, in the sense that \aobservables nearly satisfy the canonical anti-commutation relations on $supp_A \ketpsi$, and with $B_{ij} \approx \left( (A_i \pm A_j)/ \sqrt{2} \right)^T$ on $supp_B \ketpsi$. Unfortunately, that turns out not to be the case; the obstacle is that one can take one of the optimal strategies described in Theorem \ref{thm:chshn_optimal_strategies} where some blocks of the Schmidt coefficients for $\ketpsi$ are arbitrarily small, and then one can change the corresponding blocks of the observables $A_i, B_{jk}$ to something arbitrary. The result is that one gets an $\epsilon$-optimal \chshn strategy such that the observables \aobservables are not well-behaved on all of $supp_A \ketpsi$ and the observables \bobservables are not well-behaved on all of $supp_B \ketpsi$. 

The next best thing one could hope for is that the observables \aobservables, \bobservables, are well-behaved on some subspace of $supp_A \ketpsi \otimes supp_B \ketpsi$. One approach to finding such a subspace is to take a subspace of $supp_A \ketpsi$ on the $A$ side, and a subspace of $supp_B \ketpsi$ on the $B$ side. This approach has been pursued in reference \cite{slofstra2011lower}. The difficulty with this approach is that it gives error bounds that depend on the dimensions $d_A, d_B$ of the strategy. We have seen in Theorem \ref{thm:chshn_optimal_strategies} that $d_A, d_B$ can be arbitrarily large even for optimal strategies. 

In this paper, we take a different approach. We start with a strategy \abpsi on $\C^{d_A} \otimes \C^{d_B}$ that is $\epsilon$-optimal for \chshn. We introduce a new strategy \tildeabpsi on $\C^{2^{\lceil n/2 \rceil}} \otimes \C^{2^{\lceil n/2 \rceil}}$ that we call the canonical optimal strategy for \chshn. Then we construct a non-zero linear operator $T: \C^{2^{\lceil n/2 \rceil}} \otimes \C^{2^{\lceil n/2 \rceil}} \longrightarrow \C^{d_A} \otimes \C^{d_B}$ that approximately satisfies the intertwining operator property from representation theory. Formally, we prove the following:

\begin{theorem} \label{thm:chshn_near_optimal_strategies}
Let \abpsi be an $\epsilon$-optimal \chshn strategy on $\C^{d_A} \otimes \C^{d_B}$. Let \tildeabpsi be the canonical optimal strategy on $\C^{2^{\lceil n/2 \rceil}} \otimes \C^{2^{\lceil n/2 \rceil}}$. Then, there exists a non-zero linear operator 
\[T: \C^{2^{\lceil n/2 \rceil}} \otimes \C^{2^{\lceil n/2 \rceil}} \longrightarrow \C^{d_A} \otimes \C^{d_B}\] with the properties 
\begin{align*}
\forall i \quad \|(A_i \otimes I) T - T (\tilde{A}_i \otimes I) \|_F &< 12 n^2 \sqrt{\epsilon} \|T\|_F \\
\forall j \neq k \quad  \| (I \otimes B_{jk}) T - T (\tbjk) \|_F &< 17 n^2 \sqrt{\epsilon} \|T\|_F
\end{align*}
\end{theorem} 

We now define the canonical optimal strategies that are used in the statement of Theorem \ref{thm:chshn_near_optimal_strategies}. The canonical strategy is defined differently for the cases $n=2k$ and $n=2k+1$:
\begin{enumerate}
\item For the case $n=2k$ we define the canonical strategy on the space $\C^{2^k} \otimes \C^{2^k}$ to be as follows
\begin{align*}
\tilde{A}_i &= \sigma_{k,i}, \quad i = 1, \dots 2k \\
\tilde{B}_{jl} = \oneoverroottwo (\tilde{A}_j^T + \tilde{A}_l^T), &\quad \tilde{B}_{lj} = \oneoverroottwo (\tilde{A}_j^T - \tilde{A}_l^T), \quad 1 \leq j < l \leq 2k \\
\tketpsi &= \frac{1}{\sqrt{2^k}} \sum_{i=1}^{2^k} \ket{i} \otimes \ket{i}
\end{align*}
\item For the case $n=2k+1$ we define the canonical strategy on the space $\C^{2^{k+1}} \otimes \C^{2^{k+1}}$ to be as follows
\begin{align*}
\tilde{A}_i &= \begin{bmatrix}
\sigma_{k,i} & 0 \\
0 & \sigma_{k,i}
\end{bmatrix}, \quad i = 1, \dots 2k, \quad \tilde{A}_{2k+1} = \begin{bmatrix}
\sigma_{k,2k+1} & 0 \\
0 & -\sigma_{k,2k+1}
\end{bmatrix} \\
\tilde{B}_{jl} &= \oneoverroottwo (\tilde{A}_j^T + \tilde{A}_l^T), \quad \tilde{B}_{lj} = \oneoverroottwo (\tilde{A}_j^T - \tilde{A}_l^T), \quad 1 \leq j < l \leq 2k+1 \\
\tketpsi &= \frac{1}{\sqrt{2^{k+1}}} \sum_{i=1}^{2^{k+1}} \ket{i} \otimes \ket{i}
\end{align*}
\end{enumerate}
The motivation for defining the canonical strategies in this way is that the observables $A_1, \dots A_n$ generate an algebra that is isomorphic to the Clifford algebra with $n$ generators. 

Next, we say a few words about the motivation for proving a result of the form of Theorem \ref{thm:chshn_near_optimal_strategies}. We look at it from two different points of view: the point of view of the concept of homomorphism in algebra, and the point of view of identifying a "good subspace" on which the observables from a strategy are "well-behaved". 

Consider the concept of homomorphism in algebra. When we talk of a homomorphism, we have two sets with certain operations on each, and the homomorphism is a map from one set to the other that preserves all the operations. In the context of Theorem \ref{thm:chshn_near_optimal_strategies}, the two sets are  $\C^{2^{\lceil n/2 \rceil}} \otimes \C^{2^{\lceil n/2 \rceil}}$ and $\C^{d_A} \otimes \C^{d_B}$. The operations on $\C^{2^{\lceil n/2 \rceil}} \otimes \C^{2^{\lceil n/2 \rceil}}$ are addition, scalar multiplication, and the action of the operators $\tilde{A}_i \otimes I$, $I \otimes \tilde{B}_{jk}$. The operations on $\C^{d_A} \otimes \C^{d_B}$ are addition, scalar multiplication, and the action of the operators $A_i \otimes I$, $I \otimes B_{jk}$. The operator $T$ that we construct in Theorem \ref{thm:chshn_near_optimal_strategies} is linear, so it preserves addition and scalar multiplication, and it satisfies the approximate intertwining property, so it approximately maps the action of the operators $\tilde{A}_i \otimes I$, $I \otimes \tilde{B}_{jk}$ to the action of the operators $A_i \otimes I$, $I \otimes B_{jk}$. 

Next we look at Theorem \ref{thm:chshn_near_optimal_strategies} from the point of view of identifying a "good subspace" on which the observables from a strategy are "well-behaved". We mentioned above that we can think about the classification theorem for optimal \chshn strategies as saying that $supp_A \ketpsi \otimes supp_B \ketpsi \subseteq \C^{d_A} \otimes \C^{d_B}$ is a "good subspace" on which the observables from the strategy are "well-behaved". We also saw that trying to generalize this to nearly-optimal strategies encounters difficulties if we look for a good subspace of the form $V \otimes W$ with $V \subseteq supp_A \ketpsi$ and $W \subseteq supp_B \ketpsi$. 

At this point, we take a step back to the optimal \chshn strategies. We notice that for an optimal strategy, inside the space $supp_A \ketpsi \otimes supp_B \ketpsi$ there is another space: \[ span \left\{ \aj \otimes I \ketpsi \: : \: \jinzeroonetothen \right\} \] and that this space is invariant under $A_i \otimes I$, $I \otimes B_{jk}$. The motivation for looking at this space comes from the well-known relations that connect the Bell states on two qubits and the canonical optimal CHSH(2) strategy: 
\begin{align*} I \otimes I \frac{\ket{00} + \ket{11}}{\sqrt{2}} &=  \frac{\ket{00} + \ket{11}}{\sqrt{2}} \quad \quad  \sigma_x \otimes I \frac{\ket{00} + \ket{11}}{\sqrt{2}} =  \frac{\ket{10} + \ket{01}}{\sqrt{2}} \\ \sigma_z \otimes I \frac{\ket{00} + \ket{11}}{\sqrt{2}} &=  \frac{\ket{00} - \ket{11}}{\sqrt{2}} \quad \quad \sigma_x \sigma_z \otimes I \frac{\ket{00} + \ket{11}}{\sqrt{2}} =  \frac{\ket{10} - \ket{01}}{\sqrt{2}} 
\end{align*} 

When we go to the nearly-optimal \chshn strategies, it is the space \[ span \left\{ \aj \otimes I \ketpsi \: : \: \jinzeroonetothen \right\} \] that we can identify as approximately a "good subspace". It will be clear from the proof of Theorem \ref{thm:chshn_near_optimal_strategies} that for the approximate intertwining operator $T$ we construct, \[ ImT = span \left\{ \aj \otimes I \ketpsi \: : \: \jinzeroonetothen \right\} \] It is also the case that for many optimal \chshn strategies, the space \[ span \left\{ \aj \otimes I \ketpsi \: : \: \jinzeroonetothen \right\} \] cannot be written in the form $V \otimes W$\footnote{The simplest example when $span \left\{ \aj \otimes I \ketpsi \: : \: \jinzeroonetothen \right\} $ cannot be written in the form $V \otimes W$ is when Alice and Bob share two EPR pairs and use the first one for an optimal CHSH strategy. } and this is why this subspace cannot be found by methods looking for the "good subspace" of the form $V \otimes W$ with $V \subseteq supp_A \ketpsi$ and $W \subseteq supp_B \ketpsi$.

The proof of Theorem \ref{thm:chshn_near_optimal_strategies} is in Section \ref{sec:chshn_near_optimal_strategies}. The proof gives an explicit construction of the approximately intertwining operator $T$. The construction is motivated by the above insight about the importance of the space \[ span \left\{ \aj \otimes I \ketpsi \: : \: \jinzeroonetothen \right\} \] and by the group averaging technique--a common technique of constructing intertwining operators in representation theory. 

\section{Relations for optimal and nearly-optimal quantum strategies}\label{ch:relations_for_strategies}

The goal of this section is to prove Theorems \ref{thm:equality_conditions} and \ref{thm:chshn_equality_conditions}. In subsection \ref{sec:Tsirelson_inequalities_and_SDP} we explain the relationship between non-local XOR games and semi-definite programs. This relationship has been noted previously in \cite{tsirel1987quantum,cleve2004consequences}. In subsection \ref{sec:proof_idea_for_equality_conditions} we give the main idea of the proof of Theorem \ref{thm:equality_conditions}. In subsection \ref{sec:decompositions_of_dual_solution} we show how to obtain the vectors $u_1, \dots u_r, v_1, \dots v_r$ for the statement of Theorem \ref{thm:equality_conditions} from the solution to the dual semi-definite program, and we show some properties of these vectors. In subsection \ref{sec:useful_identity_and_proof} we prove a useful identity, and obtain Theorem \ref{thm:equality_conditions} as a corollary. In subsection \ref{sec:chshn_equality_conditions}, we specialize the methods from the general case to the case of the \chshn games, and we prove Theorem \ref{thm:chshn_equality_conditions}. 

\subsection{Non-local XOR games and semi-definite programs} \label{sec:Tsirelson_inequalities_and_SDP} 

Consider the maximization problem: 
\begin{equation} \label{eq:optimizing_strategy}
 \sup_{A_i, B_j, \ketpsi} \quad \sum_{i=1}^{n} \sum_{j=1}^{m} G_{ij} \brapsi A_i \otimes B_j \ketpsi 
\end{equation}
This maximization problem expresses the search for the optimal strategy for the non-local XOR game given by the $n \times m$ matrix $G$.  The supremum is taken over all valid quantum strategies for $G$. The value of the supremum, $\beta(G)$, is the quantum success bias for the game. 

We now introduce a semi-definite program: 
\begin{equation} \label{eq:primal_sdp} 
\sup_{Z \succeq 0, \: Z \cdot E_{ii} = 1, \, i=1, \dots (n+m)} \quad G_{sym} \cdot Z 
\end{equation} 
Here, $E_{ii}$ is the $(n+m) \times (n+m)$ matrix with 1 in the $i$-th diagonal entry and 0 everywhere else, and $G_{sym}$ is the $(n+m) \times (n+m)$ matrix with block form 
\[ G_{sym} = \frac{1}{2} \begin{bmatrix} 0 & G \\ G^T & 0 \end{bmatrix} \] We can think of $G_{sym}$ as the symmetric version of the game matrix $G$. 

The two maximization problems \eqref{eq:optimizing_strategy} and \eqref{eq:primal_sdp} are related as follows: for each feasible solution of one of them, there is a feasible solution of the other that achieves the same value. Formally: 

\begin{theorem} \label{thm:strategies_and_primal_sdp}
\begin{enumerate}
\item For each quantum strategy $A_i, B_j, \ketpsi$, there is an $(n+m) \times (n+m)$ matrix $Z$ that is feasible for the semi-definite program \eqref{eq:primal_sdp} and such that 
\[ G_{sym} \cdot Z = \sum_{i=1}^{n} \sum_{j=1}^{m} G_{ij} \brapsi A_i \otimes B_j \ketpsi \]
\item For each $(n+m) \times (n+m)$ matrix $Z$ that is feasible for the semi-definite program \eqref{eq:primal_sdp} there is a quantum strategy $A_i, B_j, \ketpsi$ on $\C^{2^{\lceil (n+m)/2 \rceil}} \otimes \C^{2^{\lceil (n+m)/2 \rceil}}$ such that 
\[ \sum_{i=1}^{n} \sum_{j=1}^{m} G_{ij} \brapsi A_i \otimes B_j \ketpsi = G_{sym} \cdot Z \]
\end{enumerate}
\end{theorem}

Theorem \ref{thm:strategies_and_primal_sdp} has been proved in reference \cite{tsirel1987quantum}. The exposition there uses different language, but can be converted to the language of semi-definite programs as in Theorem \ref{thm:strategies_and_primal_sdp}. The conversion to semi-definite program language has been noted in reference \cite{cleve2004consequences}.

Having established the relation between the optimization problem \eqref{eq:optimizing_strategy} and the semi-definite program \eqref{eq:primal_sdp}, we now turn attention to the dual semidefinite program. The dual to \eqref{eq:primal_sdp} is: 

\begin{equation} \label{eq:dual_sdp}
\inf_{\sum_{i=1}^{m+n} y_i E_{ii} \succeq G_{sym}} \quad \sum_{i=1}^{m+n} y_i
\end{equation}

Both the primal and the dual semi-definite programs have strictly feasible solutions; therefore, by Theorem \ref{thm:properties_of_sdp} the primal supremum is attained, the dual infimum is attained, and both are equal. Combining this with Theorem \ref{thm:strategies_and_primal_sdp}, we get that $\beta(G) = v_{primal} = v_{dual}$ and that there exists a quantum strategy that attains $\beta(G)$. 

\subsection{Proof idea for Theorem \ref{thm:equality_conditions} } \label{sec:proof_idea_for_equality_conditions} 
We are now in a position to show how to use the dual semi-definite program \eqref{eq:dual_sdp} to obtain relations that any optimal or nearly optimal quantum strategy must satisfy. 

The basic idea of the argument is to look at the duality gap and at an approximate version of the complementary slackness condition: if $y_1, \dots y_{m+n}$ is dual optimal and if $v_{primal} (1 -\epsilon) \leq G_{sym} \cdot Z \leq v_{primal}$, then
\[ v_{primal} \epsilon \geq \left(\sum_{i=1}^{m+n} y_i E_{ii} -G_{sym} \right) \cdot Z \geq 0 \]
so we can use the dual optimal solution to obtain relations on primal optimal and near-optimal solutions. We proceed with the details in the sections below. 

\subsection{Decompositions of the dual optimal solution} \label{sec:decompositions_of_dual_solution}
In the statement of Theorem \ref{thm:equality_conditions} we use vectors $u_1, \dots u_r \in \R^n$, $v_1, \dots v_r \in \R^m$. We now show how to obtain these vectors from the dual optimal solution; the argument is contained in the following lemma and its proof. 

\begin{lemma} \label{lemma:dual_vector_decomposition_properties}
Let $y_1, \dots y_{m+n}$ be an optimal solution for the dual semi-definite program \eqref{eq:dual_sdp}. Then, there exist vectors $u_1, \dots u_r \in \R^n$, $v_1, \dots v_r \in \R^m$ with the properties 
\begin{align} \label{eq:dual_vector_decomposition_properties}
&\sum_{i=1}^{r} u_i u_i^T = \sum_{i=1}^{n} y_i E_{ii} = Diag(y_1, \dots y_n) \nonumber \\
&\sum_{i=1}^{r} v_i v_i^T = \sum_{i=1}^{m} y_{n+i} E_{ii} = Diag(y_{n+1}, \dots y_{n+m}) \\
&\sum_{i=1}^{r} u_i v_i^T = \frac{1}{2} G \nonumber
\end{align}
\end{lemma} 

\begin{proof}
We look at the $(n+m) \times (n+m)$ matrix $\sum_{i=1}^{m+n} y_i E_{ii} - G_{sym}$. It is positive semi-definite by the dual constraint. Therefore, there exist vectors $w_1, \dots w_r \in \R^{m+n}$ such that 
\[ \sum_{i=1}^{m+n} y_i E_{ii} - G_{sym} = \sum_{i=1}^{r} w_i w_i^T \]
One possible such decomposition comes from the orthonormal eigenvectors of $\sum_{i=1}^{m+n} y_i E_{ii} -G_{sym}$, each eigenvector multiplied by the square root of the corresponding eigenvalue. There is also freedom in choosing this decomposition; we make a remark about this after the end of the proof. 

Now we look at the block decomposition of the matrix $\sum_{i=1}^{m+n} y_i E_{ii} - G_{sym}$ and of the vectors $w_1, \dots w_r$. The $(n+m) \times (n+m)$ matrix $\sum_{i=1}^{m+n} y_i E_{ii} - G_{sym}$ can be written in block form as 
\[ \sum_{i=1}^{m+n} y_i E_{ii} - G_{sym} = \begin{bmatrix} Diag(y_1, \dots y_n) & -G/2 \\ -G^T/2 & Diag(y_{n+1}, \dots y_{n+m}) \end{bmatrix} \]
For the vectors $w_1, \dots w_r \in \R^{n+m}$, let $u_1, \dots u_r \in \R^n$, $v_1, \dots v_r \in \R^m$ be such that 
\[ w_i = \begin{bmatrix} u_i \\ -v_i \end{bmatrix} \quad i = 1, \dots r\] 
in block form. 

By using the block decompositions, we get
\[ \begin{bmatrix} Diag(y_1, \dots y_n) & -G/2 \\ -G^T/2 & Diag(y_{n+1}, \dots y_{n+m}) \end{bmatrix} = \sum_{i=1}^{r} \begin{bmatrix} u_i \\ -v_i \end{bmatrix} \begin{bmatrix} u_i^T & -v_i^T \end{bmatrix} \] 
and from here we get the relations \eqref{eq:dual_vector_decomposition_properties}. The lemma is proved. 
\end{proof}

We remark here that the choice of decomposition \[ \sum_{i=1}^{m+n} y_i E_{ii} -G_{sym} = \sum_{i=1}^{r} w_i w_i^T \] is not unique; see for example \cite{nielsen2010quantum}[p.~103-104]. The different decompositions give rise to equivalent sets of relations; nevertheless, it will be convenient in future arguments to be able to use more than one set of relations.

\subsection{A useful identity and the proof of Theorem \ref{thm:equality_conditions}} \label{sec:useful_identity_and_proof}

So far, we have obtained the vectors $u_1, \dots u_r \in \R^n$, $v_1, \dots v_r \in \R^m$ as in Lemma \ref{lemma:dual_vector_decomposition_properties}. To complete the proof of Theorem \ref{thm:equality_conditions}, we use the following identity:

\begin{lemma} \label{lemma:useful_identity}
Let $A_1 \dots A_n$, $B_1, \dots B_m$, $\ketpsi$ be a quantum strategy. Let $u_1, \dots u_r \in \R^n$, $v_1, \dots v_r \in \R^m$ be vectors satisfying the relations \eqref{eq:dual_vector_decomposition_properties}
Then, the following identity holds:
\begin{equation} \label{eq:useful_identity}
\sum_{k=1}^{r} \left\| u_k \cdot \vec{A} \otimes I \ketpsi - I \otimes v_k \cdot \vec{B} \ketpsi \right\|^2 = \sum_{i=1}^{m+n} y_i - \sum_{i=1}^{n} \sum_{j=1}^{m} G_{ij} \brapsi A_i \otimes B_j \ketpsi
\end{equation}
\end{lemma}

\begin{proof}
We open the squares on the left-hand side:
\begin{multline*}
\sum_{k=1}^{r} \left\| u_k \cdot \vec{A} \otimes I \ketpsi - I \otimes v_k \cdot \vec{B} \ketpsi \right\|^2  \\
= \sum_{i=1}^{r} \brapsi \left(u_i \cdot \vec{A} \right)^2 \otimes I \ketpsi + \sum_{i=1}^{r} \brapsi I \otimes \left(v_i \cdot \vec{B} \right)^2 \ketpsi - 2 \sum_{i=1}^{r} \brapsi \left(u_i \cdot \vec{A} \right) \otimes \left(v_i \cdot \vec{B} \right) \ketpsi 
\end{multline*}

Now, from the property
\[ \sum_{i=1}^{r} u_i u_i^T = \sum_{i=1}^{n} y_i E_{ii} \]
we obtain 
\[ \sum_{i=1}^{r} \left(u_i \cdot \vec{A} \right)^2 = \sum_{i=1}^n y_i A_i^2 + \sum_{i \neq j} 0 A_i A_j = \left( \sum_{i=1}^n y_i \right) I \]

Similarly, from the property
\[ \sum_{i=1}^{r} v_i v_i^T = \sum_{i=1}^{m} y_{n+i} E_{ii} \]
we obtain 
\[ \sum_{i=1}^{r} \left(v_i \cdot \vec{B} \right)^2 = \sum_{i=1}^m y_{n+i} B_i^2 + \sum_{i \neq j} 0 B_i B_j = \left( \sum_{i=1}^m y_{n+i} \right) I \]

Finally, from the property 
\[ \sum_{i=1}^{r} u_i v_i^T = \frac{1}{2} G \]
we obtain 
\[ 2 \sum_{i=1}^{r} \left(u_i \cdot \vec{A} \right) \otimes \left(v_i \cdot \vec{B} \right) = \sum_{i=1}^{n} \sum_{j=1}^{m} G_{ij} A_i \otimes B_j \]

The identity \eqref{eq:useful_identity} follows. 
\end{proof}

Using Lemma \ref{lemma:useful_identity}, we can complete the proof of Theorem \ref{thm:equality_conditions}. 

\begin{proof}[Proof of Theorem \ref{thm:equality_conditions}]
We have chosen $y_1, \dots y_{n+m}$ to be a dual optimal solution, so $\sum_{i=1}^{n+m} y_i = \beta(G)$. Then, by Lemma \ref{lemma:useful_identity}, 
\[ \sum_{k=1}^{r} \left\| u_k \cdot \vec{A} \otimes I \ketpsi - I \otimes v_k \cdot \vec{B} \ketpsi \right\|^2 = \beta(G) - \sum_{i=1}^{n} \sum_{j=1}^{m} G_{ij} \brapsi A_i \otimes B_j \ketpsi \] 

It follows that \[ (1-\epsilon) \beta(G)  \leq \sum_{i=1}^{n} \sum_{j=1}^{m} G_{ij} \brapsi A_i \otimes B_j \ketpsi \leq \beta(G)\] if and only if 
\[ \sum_{k=1}^{r} \left\| u_k \cdot \vec{A} \otimes I \ketpsi - I \otimes v_k \cdot \vec{B} \ketpsi \right\|^2 \leq \beta(G) \epsilon \] Theorem \ref{thm:equality_conditions} is proved. 
\end{proof}

\subsection{Relations for \chshn optimal and nearly optimal strategies} \label{sec:chshn_equality_conditions}

In this section, we prove Theorem \ref{thm:chshn_equality_conditions}. We look at the dual semi-definite program corresponding to the \chshn game, and we find two explicit decompositions of the form given in subsection \ref{sec:decompositions_of_dual_solution}. Using these decompositions, we obtain Theorem \ref{thm:chshn_equality_conditions}. 

We take the $n \times n(n-1)$ matrix $G$ that summarizes the information for the \chshn game. From subsection \ref{subsec:definition_of_chshn} we know that \[ G = \frac{1}{4{n \choose 2}} \sum_{1 \leq i < j \leq n} \left( \ket{i} \bra{ij} + \ket{j} \bra{ij} + \ket{i} \bra{ji} - \ket{j} \bra{ji} \right)\]

Next, we form the $n^2 \times n^2$ matrix $G_{sym}$ which has the block form:
\[ G_{sym} =  \frac{1}{2} \begin{bmatrix} 0 & G \\ G^T & 0 \end{bmatrix} \]
In this context, it is convenient to think of $\R^{n^2}$ as having an orthonormal basis formed by concatenating the basis $\ket{1}, \dots \ket{n}$ of $\R^n$ and the basis $\ket{ij}, \: i \neq j \in \{1, \dots n\}$ of $\R^{n(n-1)}$. So, we can write 
\begin{multline*}
G_{sym} = \frac{1}{8 {n \choose 2}}  \sum_{1 \leq i < j \leq n} \Big( \ket{i} \bra{ij} + \ket{j} \bra{ij} + \ket{i} \bra{ji} - \ket{j} \bra{ji} \\
+ \ket{ij} \bra{i} + \ket{ij} \bra{j} +\ket{ji} \bra{i} - \ket{ji} \bra{j}  \Big) 
\end{multline*}

Next, we form the dual semi-definite program corresponding to the  \chshn game; it is
\begin{align*}
\inf_{\sum_{i=1}^{n^2} y_i E_{ii} \succeq G_{sym}} \quad \sum_{i=1}^{n^2} y_i
\end{align*}

We know that the optimal value is $\oneoverroottwo$; this follows from the result in reference \cite{slofstra2011lower} about the quantum success bias of the \chshn game, and the discussion in Section \ref{sec:Tsirelson_inequalities_and_SDP}. 

Next, we claim that $y_1 = \dots = y_n = \frac{1}{2 \sqrt{2} n}$, $y_{n+1} = \dots = y_{n^2} = \frac{1}{2 \sqrt{2} n (n-1)}$ is a dual optimal solution. We can see that $\sum_{i=1}^{n^2} y_i = \oneoverroottwo$, the dual optimum, so all that is left to prove is that $y_1, \dots y_{n^2}$ is dual feasible. 

To prove that $y_1, \dots y_{n^2}$ is dual feasible, we show that 
\[ \sum_{i=1}^{n^2} y_i E_{ii} \succeq G_{sym} \]

We define the following vectors for $1 \leq i <j \leq N$
\begin{align*} 
& u_{ij} = \ket{i} \quad \quad  v_{ij} = \frac{\ket{ij} + \ket{ji}}{\sqrt{2}} \nonumber \\
& u_{ji} = \ket{j} \quad \quad  v_{ji} = \frac{\ket{ij} - \ket{ji}}{\sqrt{2}} 
\end{align*}
and observe that the following decomposition holds:
\begin{multline}\label{eq:chshn_dual_decomposition1}
\sum_{i=1}^{n^2}  y_i E_{ii} - G_{sym}  \\
= \frac{1}{2 \sqrt{2} n (n-1)} \sum_{1 \leq i < j \leq n} \left( \left( u_{ij} - v_{ij} \right) \left( u_{ij} - v_{ij} \right)^T + \left( u_{ji} - v_{ji} \right) \left( u_{ji} - v_{ji} \right)^T \right)
\end{multline}

It follows that the matrix $\sum_{i=1}^{n^2}  y_i E_{ii} - G_{sym}$ is positive semi-definite, and therefore, the given $y_1, \dots y_{n^2}$ are a dual optimal solution as claimed. 

Now, from the decomposition \eqref{eq:chshn_dual_decomposition1}, we conclude that the following two statements are equivalent: 
\begin{itemize}
\item $A_i, B_{jk}, \ketpsi$ is an $\epsilon$-optimal CHSH($n$) strategy.
\item The observables and state satisfy
\begin{multline*}
\sum_{1 \leq i < j \leq n} \Bigg( \left\| A_i \otimes I \ketpsi - I \otimes \frac{B_{ij} + B_{ji}}{\sqrt{2}} \ketpsi \right\|^2  \\
+ \left\| A_j \otimes I \ketpsi - I \otimes \frac{B_{ij} - B_{ji}}{\sqrt{2}} \ketpsi \right\|^2 \Bigg) \leq 2 n (n-1) \epsilon 
\end{multline*}
\end{itemize}
The argument is the same as the argument in subsections \ref{sec:decompositions_of_dual_solution} and \ref{sec:useful_identity_and_proof}. 

Next, we define the following vectors for $1 \leq i <j \leq N$
\begin{align*} 
& u'_{ij} = \frac{\ket{i} + \ket{j}}{\sqrt{2}} \quad \quad v'_{ij} = \ket{ij} \nonumber \\
& u'_{ji} = \frac{\ket{i} - \ket{j}}{\sqrt{2}} \quad \quad v'_{ji} = \ket{ji}
\end{align*}
and observe that the following decomposition holds:
\begin{multline*}
\sum_{i=1}^{n^2}  y_i E_{ii} - G_{sym} \\
= \frac{1}{2 \sqrt{2} n (n-1)} \sum_{1 \leq i < j \leq n} \left( \left( u'_{ij} - v'_{ij} \right) \left( u'_{ij} - v'_{ij} \right)^T + \left( u'_{ji} - v'_{ji} \right) \left( u'_{ji} - v'_{ji} \right)^T \right)
\end{multline*}

From this we conclude that the following two statements are equivalent:
\begin{itemize}
\item $A_i, B_{jk}, \ketpsi$ is an $\epsilon$-optimal CHSH($n$) strategy. 
\item The observables and state satisfy
\begin{multline*} 
\sum_{1 \leq i < j \leq n} \Bigg( \left\| \frac{A_i + A_j}{\sqrt{2}} \otimes I \ketpsi - I \otimes B_{ij} \ketpsi \right\|^2  \\ 
+ \left\| \frac{A_i - A_j}{\sqrt{2}} \otimes I \ketpsi - I \otimes B_{ji} \ketpsi \right\|^2 \Bigg) \leq 2 n (n-1) \epsilon 
\end{multline*}
\end{itemize}

This completes the proof of Theorem \ref{thm:chshn_equality_conditions}. 

\section{Classification of \chshn optimal strategies} \label{sec:chshn_optimal_strategies}

The goal of this section is to prove Theorem \ref{thm:chshn_optimal_strategies}. Theorem \ref{thm:chshn_optimal_strategies} claims the equivalence of two statements: 
\begin{itemize}
\item A strategy is optimal for the \chshn game
\item There are bases for Alice's space and for Bob's space with respect to which the strategy has a certain form. 
\end{itemize}

We prove that the first statement implies the second in subsection \ref{subsec:optimal_chshn_strategy_must_have_a_certain_form}, and we prove that the second statement implies the first in subsection \ref{subsec:a_strategy_of_a_certain_form_is_optimal_for_chshn}.

\subsection{An optimal \chshn strategy must have a certain form}\label{subsec:optimal_chshn_strategy_must_have_a_certain_form}

Let \abpsi be an arbitrary optimal \chshn strategy on $\C^{d_A} \otimes \C^{d_B}$. Our goal is to show that this strategy has the structure described in Theorem \ref{thm:chshn_optimal_strategies}. 

From Theorem \ref{thm:chshn_equality_conditions} we know that the following relations are satisfied for all $i, j \: 1 \leq i < j \leq n$
\begin{align*} 
A_i \otimes I \ketpsi &= I \otimes \frac{B_{ij} + B_{ji}}{\sqrt{2}} \ketpsi \\
A_j \otimes I \ketpsi &= I \otimes \frac{B_{ij} - B_{ji}}{\sqrt{2}} \ketpsi
\end{align*} 

Let $\Psi = \mathcal{L}(\ketpsi)$ be the $d_A \times d_B$ matrix that corresponds to $\ketpsi \in \C^{d_A} \otimes \C^{d_B}$ (subsection \ref{subsec:a_linear_bijection}). To the relations above correspond the following relations in terms of $\Psi$:
\begin{equation} \label{eq:intertwining relations}
\begin{aligned}
A_i \Psi &= \Psi \left(\frac{B_{ij} + B_{ji}}{\sqrt{2}}\right)^T \\
A_j \Psi &= \Psi \left(\frac{B_{ij} - B_{ji}}{\sqrt{2}} \right)^T
\end{aligned}
\end{equation}

It follows that the space $Im \Psi \subset \C^{d_A}$ is invariant under the observables \aobservables, by using Schur's Lemma (lemma \ref{lemma:Schurs_lemma}). 

Let the non-zero terms in the Schmidt decomposition of $\ketpsi$ be 
\[ \ketpsi = \sum_{i=1}^{r} \sqrt{\lambda_i} \ket{u_i} \otimes \ket{v_i} \]
Choose $\ket{u_1}, \dots \ket{u_r}$ as an orthonormal basis of $Im \Psi$, and complete it to an orthonormal basis of $\C^{d_A}$. With respect to this basis, the observables \aobservables have the block form 
\[A_i = \begin{bmatrix}
A'_i & 0 \\
0 & C_i
\end{bmatrix} \]
where $A'_i$ acts on $Im \Psi$, and $C_i$ acts on the orthgonal complement. 

From $A_i^\dagger = A_i, \: A_i^2 = I$, it follows that ${A'_i}^\dagger = A'_i, \: {A'_i}^2 = I$ and $C_i^\dagger = C_i, \: C_i^2 = I$. 

It is clear at this point that the blocks $C_i, \: i=1, \dots n$ may be arbitrary, and that they don't in any way influence the quantum value achieved by the strategy. From now on, we focus on the observables $A'_i$ that act on the space $Im \Psi$. 

We now claim that for all $i,j, \: 1 \leq i <j \leq n$, $\{A'_i, A'_j\}=0$. This is because 
\[ \{A_i, A_j\} \Psi = \Psi \left\{ \frac{B_{ij}^T + B_{ji}^T}{\sqrt{2}}, \frac{B_{ij}^T - B_{ji}^T}{\sqrt{2}} \right\} = \Psi \left( (B_{ij}^T)^2 - (B_{ji}^T)^2 \right) = 0 \]

It follows that $A'_1, \dots A'_n$ are anti-commuting $\pm 1$ observables on the space $Im \Psi$. We apply Theorem \ref{thm:the_general_form_of_n_anti_commuting_observables} and get that the number of non-zero Schmidt coefficients of $\ketpsi$ is an integer multiple of $\poweroftwo$. Let $r = s \poweroftwo$. 

We now consider the operator $\Psi \Psi^\dagger$, which takes the space $Im \Psi$ to itself. Form the relations \eqref{eq:intertwining relations}, it follows that 
\[ A_i \Psi \Psi^\dagger = \Psi \left(\frac{B_{ij} + B_{ji}}{\sqrt{2}}\right)^T \Psi^\dagger = \Psi \left(\left(\frac{B_{ij} + B_{ji}}{\sqrt{2}}\right)^T \right)^\dagger \Psi^\dagger = \Psi \Psi^\dagger A_i^\dagger = \Psi \Psi^\dagger A_i\]

We now apply Schur's lemma, and conclude that all eigenspaces of $\Psi \Psi^\dagger$ must be invariant spaces for the observables $A'_1, \dots A'_n$. It then follows that all eigenspaces of $\Psi \Psi^\dagger$ must have dimension an integer multiple of $\poweroftwo$. 

From this conclusion about the eigenspaces of $\Psi \Psi^\dagger$, and from the expression \[ \Psi \Psi^\dagger = \sum_{i=1}^{s \poweroftwo} \lambda_i \ket{u_i} \bra{u_i} \] we get that the non-zero Schmidt coefficients of $\ketpsi$ must come in blocks of length $\poweroftwo$ that are equal, i.e.
\begin{align*}
\lambda_1 = &\dots = \lambda_{\poweroftwo} \\
\lambda_{\poweroftwo +1} = &\dots = \lambda_{2 \cdot \poweroftwo} \\
&\dots \\
\lambda_{(s-1) \poweroftwo +1} = &\dots = \lambda_{s \poweroftwo}
\end{align*}

Returning to the observables $A'_1, \dots A'_n$, we apply Theorem \ref{thm:the_general_form_of_n_anti_commuting_observables} and get that with respect to the basis $\ket{u_1}, \dots \ket{u_{s \poweroftwo}}$, the observables $A'_1, \dots A'_n$ have the block diagonal form:
\[ A'_i = \begin{bmatrix}
A_i^{(1)} && \\
& \ddots & \\
&& A_i^{(s)} 
\end{bmatrix} \]
where each $A_i^{(j)}$ is $\poweroftwo \times \poweroftwo$ and acts on $span(\ket{u_{(j-1) \poweroftwo +1}}, \dots \ket{u_{j \poweroftwo}})$, and, for each $i = 1, \dots n$, for each $j=1, \dots s$, $A_i^{(j)} = \sigma_{\lfloor n/2 \rfloor, i}$ except for the case $n =2k+1$, and $i=n$, in which case $A_i^{(j)}$ is either $\sigma_{k,2k+1}$ or $-\sigma_{k,2k+1}$.

The proof of the forward direction of Theorem \ref{thm:chshn_optimal_strategies} is now almost complete; it remains to prove the statement about \bobservables. We take the following relations from Theorem \ref{thm:chshn_equality_conditions}: 
\begin{align*}
\frac{A_i + A_j}{\sqrt{2}} \otimes I \ketpsi &= I \otimes B_{ij} \ketpsi \\
\frac{A_i - A_j}{\sqrt{2}} \otimes I \ketpsi &= I \otimes B_{ji} \ketpsi
\end{align*} 
and we rewrite them in terms of $\Psi^T$ to get
\begin{equation}\label{eq:transpose_intertwining_relations}
\begin{aligned}
B_{ij} \Psi^T &= \Psi^T \left( \frac{A_i + A_j}{\sqrt{2}} \right)^T \\
B_{ji} \Psi^T &= \Psi^T \left( \frac{A_i - A_j}{\sqrt{2}} \right)^T
\end{aligned}
\end{equation}

It follows from Schur's lemma that $Im \Psi^T = span(\ket{v_1}, \dots \ket{v_{s \poweroftwo}} )$ is invariant under \bobservables, and so \bobservables have the block diagonal form
\[B_{jk} = \begin{bmatrix}
B'_{jk} & 0 \\
0 & D_{jk}
\end{bmatrix} \]
where the $\pm 1$ observables $B'_{jk}$ act on $Im \Psi^T$ and the $\pm 1$ observables $D_{jk}$ act on the orthogonal complement. 

The final thing that is left to show is the block-diagonal decomposition
\[B'_{jk} = \begin{bmatrix}
B_{jk}^{(1)} && \\
& \ddots & \\
&& B_{jk}^{(s)}
\end{bmatrix} \]
and the relations on the individual blocks, for $1 \leq j < k \leq n$
\[ B_{jk}^{(l)} = \left(\frac{A_j^{(l)} + A_k^{(l)}}{\sqrt{2}}\right)^T \quad  B_{kj}^{(l)} = \left(\frac{A_j^{(l)} - A_k^{(l)}}{\sqrt{2}}\right)^T \]
These follow from the relations \eqref{eq:transpose_intertwining_relations} and from the fact that with respect to the basis  $\ket{u_1^*}, \dots \ket{u_{s \poweroftwo}^*}$ of the source space and the basis $\ket{v_1}, \dots \ket{v_{s \poweroftwo}}$ of the target space, $\Psi^T$ has the block diagonal form 
\[ \Psi^T = \begin{bmatrix}
\sqrt{\lambda_{\poweroftwo}} I & & \\
& \ddots & \\
& & \sqrt{\lambda_{s \poweroftwo}} I
\end{bmatrix}  \]

The forward direction of Theorem \ref{thm:chshn_optimal_strategies} is proved. 

\subsection{Any strategy of a certain form is optimal for \chshn}\label{subsec:a_strategy_of_a_certain_form_is_optimal_for_chshn}

We assume that a strategy \abpsi on $\C^{d_A} \otimes \C^{d_B}$ has the form described in Theorem \ref{thm:chshn_optimal_strategies}. We have to show that \abpsi is an optimal \chshn strategy. 

First, we use the description of the Schmidt decomposition of $\ketpsi$ (the first bullet), to write
\[ \ketpsi = \sum_{l=1}^s \sqrt{\poweroftwo} \sqrt{\lambda_{l \poweroftwo}} \ket{\psi_l}  \]
where 
\[ \ket{\psi_l} = \frac{1}{\sqrt{\poweroftwo}} \sum_{r=(l-1) \poweroftwo +1}^{l \poweroftwo} \ket{u_r} \otimes \ket{v_r} \]

Next, we claim that for each $i,j, \: 1 \leq i < j \leq n$, the following two statements hold, the first for indvidual blocks, and the second for the whole observables:
\begin{itemize}
\item For each block number $l, \; 1 \leq l \leq s$, 
\begin{equation}\label{eq:blocks_satisfy_chsh_correlations}
\begin{aligned}
\bra{\psi_l} A_i^{(l)} \otimes B_{ij}^{(l)} \ket{\psi_l} &= \frac{1}{\sqrt{2}} \quad \quad
\bra{\psi_l} A_i^{(l)} \otimes B_{ji}^{(l)} \ket{\psi_l} = \frac{1}{\sqrt{2}} \\
\bra{\psi_l} A_j^{(l)} \otimes B_{ij}^{(l)} \ket{\psi_l} &= \frac{1}{\sqrt{2}} \quad \quad
\bra{\psi_l} A_j^{(l)} \otimes B_{ji}^{(l)} \ket{\psi_l} = - \frac{1}{\sqrt{2}} 
\end{aligned}
\end{equation}
\item For the whole observables, 
\begin{equation}\label{eq:observables_satisfy_chshn_correlations}
\begin{aligned}
\bra{\psi} A_i \otimes B_{ij} \ket{\psi} &= \frac{1}{\sqrt{2}} \quad \quad
\bra{\psi} A_i \otimes B_{ji} \ket{\psi} = \frac{1}{\sqrt{2}} \\
\bra{\psi} A_j \otimes B_{ij} \ket{\psi} &= \frac{1}{\sqrt{2}} \quad \quad
\bra{\psi} A_j \otimes B_{ji} \ket{\psi} = - \frac{1}{\sqrt{2}} 
\end{aligned}
\end{equation}
\end{itemize} 

Consider the first statement, the one for individual blocks. We know $A_i^{(l)}, \, i = 1, \dots n$ anti-commute on the space $span(\ket{u_{(l-1) \poweroftwo +1}}, \dots \ket{u_{l \poweroftwo}}) = supp_A \ket{\psi_l}$. We also know that on the space $span(\ket{v_{(l-1) \poweroftwo +1}}, \dots \ket{v_{l \poweroftwo}}) = supp_B \ket{\psi_l}$ we have, for $1 \leq j <k \leq n$
\[ B_{jk}^{(l)} = \left(\frac{A_j^{(l)} + A_k^{(l)}}{\sqrt{2}}\right)^T \quad  B_{kj}^{(l)} = \left(\frac{A_j^{(l)} - A_k^{(l)}}{\sqrt{2}}\right)^T \] And finally, we know $\ket{\psi_l}$ is maximally entangled on \[span(\ket{u_{(l-1) \poweroftwo +1}}, \dots \ket{u_{l \poweroftwo}}) \otimes span(\ket{v_{(l-1) \poweroftwo +1}}, \dots \ket{v_{l \poweroftwo}})\] We apply Lemma \ref{lemma:anti_commuting_observables_and_inner_products}and obtain the relations \eqref{eq:blocks_satisfy_chsh_correlations}. 

The statement for the whole observables follows from the statement for the individual blocks. We show this for $\brapsi A_i \otimes B_{ij} \ketpsi$:
\begin{multline*}
\brapsi A_i \otimes B_{ij} \ketpsi = \sum_{l=1}^s \poweroftwo \lambda_{l \poweroftwo} \bra{\psi_l} A_i^{(l)} \otimes B_{ij}^{(l)} \ket{\psi_l}  
= \sum_{l=1}^s \poweroftwo \lambda_{l \poweroftwo} \frac{1}{\sqrt{2}} = \frac{1}{\sqrt{2}}
\end{multline*}
The other three terms are analogous. 

Now, from the relations \eqref{eq:observables_satisfy_chshn_correlations}, we see that the \chshn value of the strategy \abpsi is 
\[ \chshnexpression = \frac{1}{\sqrt{2}}  \]
so \abpsi is an optimal \chshn strategy. The reverse direction of Theorem \ref{thm:chshn_optimal_strategies} is proved. 

\section{Approximate intertwining operator construction for \chshn near-optimal strategies} \label{sec:chshn_near_optimal_strategies}

The goal of this section is to prove Theorem \ref{thm:chshn_near_optimal_strategies}. That is, given an arbitrary $\epsilon$-optimal \chshn strategy \abpsi on \strategyspace, and the canonical optimal \chshn strategy \tildeabpsi on \canonicalstrategyspace, we want to show the existence of a non-zero linear operator \[T: \C^{2^{\lceil n/2 \rceil}} \otimes \C^{2^{\lceil n/2 \rceil}} \longrightarrow \C^{d_A} \otimes \C^{d_B}\] with the properties 
\begin{align*}
\forall i \quad \|(A_i \otimes I) T - T (\tilde{A}_i \otimes I) \|_F &< 12 n^2 \sqrt{\epsilon} \|T\|_F \\
\forall j \neq k \quad  \| (I \otimes B_{jk}) T - T (\tbjk) \|_F &< 17 n^2 \sqrt{\epsilon} \|T\|_F
\end{align*}

We construct $T$ explicitly:
\[ T = \intertwiningoperatorexpression \]
The motivation for this construction comes from the insight about the importance of the space \[ span \left\{ \aj \otimes I \ketpsi \: : \: \jinzeroonetothen \right\} \] and from the group averaging technique of constructing intertwining operators. In our context, representations of finite groups are not explicitly present. However, the relations on optimal and nearly-optimal \chshn strategies from Theorem \ref{thm:chshn_equality_conditions} are very strong and we can use them to prove the $T$ defined above behaves approximately like an intertwining operator with respect to the observables of the two strategies. 

The argument proceeds in the following steps:
\begin{enumerate}
\item We prove that the vectors 
\[ \left\{ \tildeajpsi \: : \: \jinzeroonetothen \right\} \] coming from the canonical strategy are orthonormal. 
\item From this, we derive that $\|T\|_F = 1$, and so also $T \neq 0$. 
\item Next, we show that we can write 
\begin{multline}
(\ai) T - T (\tai) = \avsumj \Big( A_i \ajpsi  \\
- \ajpsimodi \Big) \tildepsiajinv
\end{multline}
Here the $\signij$ notation has to do with the sign resulting from changing the order in a product of anti-commuting observables and will be defined in detail later. 
\item Similarly, we show we can write
\begin{multline}
(\bkl) T - T (\tbkl) = \avsumj \Bigg( \ajbklpsi  \\
 - \oneoverroottwo \Big( \pm \ajpsimodk  \\
 + \ajpsimodl \Big) \Bigg) 
 \tildepsiajinv
\end{multline}
In the place where there is $\pm$, we take $+$ if $k<l$ and we take $-$ if $k>l$.
\item Next, we show that for all $i \in \{1, \dots n\}$, for all $\jinzeroonetothen$, 
\begin{multline}
\Big\| A_i \ajpsi - \ajpsimodi \Big\|   \\
\leq (6 + 4 \sqrt{2}) n^2 \sqrt{\epsilon} < 12 n^2 \sqrt{\epsilon}
\end{multline}
\item Similarly we show that for all $k \neq l \in \{1, \dots n\}$, for all $\jinzeroonetothen$,
\begin{multline}
\Bigg\| \ajbklpsi  
 - \oneoverroottwo \Big( \pm \ajpsimodk \\
 + \ajpsimodl \Big) \Bigg\| \\
 \leq \left( \frac{17}{2} + 6 \sqrt{2} \right) n^2 \sqrt{\epsilon} < 17 n^2 \sqrt{\epsilon}
\end{multline}
\item Finally, we combine all the previous steps to show that 
\begin{align*}
\forall i \quad \|(A_i \otimes I) T - T (\tilde{A}_i \otimes I) \|_F &< 12 n^2 \sqrt{\epsilon} \|T\|_F \\
\forall j \neq k \quad  \| (I \otimes B_{jk}) T - T (\tbjk) \|_F &< 17 n^2 \sqrt{\epsilon} \|T\|_F
\end{align*}
as required for the proof of Theorem \ref{thm:chshn_near_optimal_strategies}. 
\end{enumerate}

The seven subsections below are devoted to the detailed arguments for the seven steps outlined above. 

\subsection{Orthonormal vectors}\label{subsec:orthonormal_vectors}

Here we aim to show that the vectors 
\[ \left\{ \tildeajpsi \: : \: \jinzeroonetothen \right\} \] coming from the canonical strategy are orthonormal. 

First, we reduce this to proving that $\ket{\tilde{\psi}}$ is orthogonal to $\tildeajpsi$ for each nonzero $\jinzeroonetothen$. This works because we can use the anti-commutation relations for the \tildeaobservables to show that given $(k_1, \dots k_n), \, (l_1, \dots l_n) \, \in \, \{0,1\}^n$, one can take $(j_1, \dots j_n) = (k_1 \oplus l_1, \dots k_n \oplus l_n)$ and have \[  \bra{\tilde{\psi}} \left( \tilde{A}_1^{k_1} \dots \tilde{A}_n^{k_n} \otimes I \right)^{\dagger} \tilde{A}_1^{l_1} \dots \tilde{A}_n^{l_n} \otimes I  \ket{\tilde{\psi}} = \bra{\tilde{\psi}} \left( \pm \tildeajpsi \right) \]

Now, we prove that $\ket{\tilde{\psi}}$ is orthogonal to $\tildeajpsi$ for each nonzero $\jinzeroonetothen$. There are two cases: one case is if $n$ is odd and $(j_1, \dots j_n) = (1, \dots 1)$ and the second case is all other situations. 

We consider the first case. For $n$ odd, we have
\[ \prod_{i=1}^n \tilde{A}_i = (-\complexi )^n \begin{bmatrix}
I & 0 \\
0 & -I
\end{bmatrix} \]

Therefore, we have 
\begin{align*}
\ket{\tilde{\psi}} = &\frac{1}{\sqrt{2 \cdot \poweroftwo}} \left( \sum_{j=1}^{\poweroftwo} \ket{j} \otimes \ket{j} +  \sum_{j=\poweroftwo + 1}^{2 \cdot \poweroftwo} \ket{j} \otimes \ket{j} \right) \\
\prod_{i=1}^n \tilde{A}_i \otimes I \ket{\tilde{\psi}} = &\frac{(-\complexi )^n}{\sqrt{2 \cdot \poweroftwo}} \left( \sum_{j=1}^{\poweroftwo} \ket{j} \otimes \ket{j} -  \sum_{j=\poweroftwo + 1}^{2 \cdot \poweroftwo} \ket{j} \otimes \ket{j} \right) \\
\end{align*}
and so $\ket{\tilde{\psi}}$ is orthogonal to $\tildeajpsi$ in the first case.

Next, we consider the second case. First, we look at the product $\tilde{A}_1^{j_1} \dots \tilde{A}_n^{j_n}$. We claim that there exists an index $i$ such that 
\begin{equation} \label{eq:conjugating_a_product}
\tilde{A}_i \tilde{A}_1^{j_1} \dots \tilde{A}_n^{j_n} \tilde{A}_i = - \tilde{A}_1^{j_1} \dots \tilde{A}_n^{j_n}
\end{equation} 
This is because when there are an even number of terms in the product $\tilde{A}_1^{j_1} \dots \tilde{A}_n^{j_n}$, we can choose $\tilde{A}_i$ to be one of the observables that appears in the product, and if there are an odd number of terms, we can choose $\tilde{A}_i$ to be one of the observables that does not appear in the product. 
Next, we use the relation \eqref{eq:conjugating_a_product} to write
\begin{multline*}
\bra{\tilde{\psi}} \tildeajpsi = \bra{\tilde{\psi}} \left( \tilde{A}_i \otimes \tilde{A}_i^T \right) \left( \tildeajpsi \right) \left( \tilde{A}_i \otimes \tilde{A}_i^T \right) \ket{\tilde{\psi}} = \\ 
= \bra{\tilde{\psi}} (\tilde{A}_i \tilde{A}_1^{j_1} \dots \tilde{A}_n^{j_n} \tilde{A}_i) \otimes (\tilde{A}_i^T)^2 \ket{\tilde{\psi}} = - \bra{\tilde{\psi}} \tildeajpsi
\end{multline*}
and from here we obtain that $\ket{\tilde{\psi}}$ is orthogonal to $\tildeajpsi$ in the second case as well. This completes the proof that the vectors \[ \left\{ \tildeajpsi \: : \: \jinzeroonetothen \right\} \] coming from the canonical strategy are orthonormal. 

\subsection{The Frobenius norm of $T$}\label{subsec:frobenius_norm_of_t}

Here we aim to prove that $\|T\|_F = 1$. This follows from the expression
\[ \|T\|_F = \sqrt{Tr \: T T^\dagger} \]
for the Frobenius norm, combined with the expression 
\[ T = \intertwiningoperatorexpression \]
for $T$, combined with the fact that the vectors 
\[ \left\{ \tildeajpsi \: : \: \jinzeroonetothen \right\} \]
are orthnormal, and combined with the fact that the vectors
\[ \left\{ \ajpsi \: : \: \jinzeroonetothen \right\} \]
all have unit norm. 

To combine all these facts, we use the following lemma:
\begin{lemma}\label{lemma:calculating_frobenius_norm}
Let \[S = \frac{1}{\sqrt{r}}\sum_{i=1}^r \ket{u_i} \bra{v_i}\] where the vectors $\ket{v_i}, \: i=1, \dots, r$ are orthnormal. Then, \[ \|S\|_F = \sqrt{\frac{\sum_{i=1}^r \|u_i\|^2}{r}} \]
\end{lemma}
\begin{proof} We know that
\[ S S^\dagger = \frac{1}{r} \sum_{i=1}^r \sum_{j=1}^r \ket{u_i} \bra{v_i}\ket{v_j} \bra{u_j} =  \frac{1}{r} \sum_{i=1}^r \ket{u_i} \bra{u_i} \]
and so 
\[ \|S\|_F  = \sqrt{\frac{\sum_{i=1}^r Tr \: \ket{u_i} \bra{u_i} }{r}} = \sqrt{\frac{\sum_{i=1}^r \|u_i\|^2}{r}} \]
\end{proof}

Applying this lemma to the operator $T$, we conclude that $\|T\|_F = 1$. 

\subsection{The expression for $(\ai) T - T (\tai)$} \label{subsec:the_expression_for_tai}

Here we aim to show the identity 
\begin{multline*}
(\ai) T - T (\tai) = \avsumj \Big( A_i \ajpsi  \\
- \ajpsimodi \Big) \tildepsiajinv
\end{multline*}

Consider $T (\tai)$:
\begin{multline*}
T (\tai) = \avsumj \ajpsi \tildepsiajinv (\tai)  \\
= \avsumj \ajpsi \Big( (\tai) (\tildeajpsi) \Big)^\dagger
\end{multline*}

We now use the anti-commutation relations for \tildeaobservables to insert the $\tilde{A}_i$ into the product $\tilde{A}_1^{j_1} \dots \tilde{A}_n^{j_n}$. This possibly incurs a minus sign, depending on the particular $i$ and the particular $\jinzeroonetothen$. We define $\signij$ to be such that 
\[ (\tilde{A}_i)(\tilde{A}_1^{j_1} \dots \tilde{A}_n^{j_n}) = \signij \tilde{A}_1^{j_1} \dots \tilde{A}_i^{j_i \oplus 1} \dots \tilde{A}_n^{j_n} \]

Using this, we get 
\begin{multline*} T (\tai) \\ = \avsumj \ajpsi 
 \Big(\signij \tilde{A}_1^{j_1} \dots \tilde{A}_i^{j_i \oplus 1} \dots \tilde{A}_n^{j_n} \otimes I \ket{\tilde{\psi}} \Big)^\dagger 
\end{multline*}

Now we change the index of summation, and use \[sign(i, j_1, \dots j_i \dots j_n) = sign(i, j_1, \dots j_i \oplus 1 \dots j_n)\] to get 
\begin{multline*}
T (\tai) \\ = \avsumj \ajpsimodi 
\tildepsiajinv
\end{multline*}

From here, the identity 
\begin{multline*}
(\ai) T - T (\tai) = \avsumj \Big( A_i \ajpsi  \\
- \ajpsimodi \Big) \tildepsiajinv
\end{multline*}
follows. 

\subsection{The expression for $(\bkl) T - T (\tbkl)$}\label{subsec:the_expression_for_tbkl}

Here we aim to prove the identity 
\begin{multline*}
(\bkl) T - T (\tbkl) = \avsumj \Bigg( \ajbklpsi  \\
 - \oneoverroottwo \Big( \pm \ajpsimodk  \\
 + \ajpsimodl \Big) \Bigg) 
 \tildepsiajinv
\end{multline*}

The argument is similar to the previous section. We consider $T(\tbkl)$. 
\begin{multline*}
T (\tbkl) 
= \avsumj \ajpsi \Big( (\tajtensori) (\tbkl) \tketpsi \Big)^\dagger  \\
= \avsumj \ajpsi \Big( (\tajtensori) (\frac{ \pm \tilde{A}_k + \tilde{A}_l }{\sqrt{2}} \otimes I) \tketpsi \Big)^\dagger
\end{multline*}
where $+\tilde{A}_k$ is taken if $k<l$ and $-\tilde{A}_k$ is taken if $k>l$. 

Next, we use the anti-commutation relations to insert $\tilde{A}_k$ and $\tilde{A}_l$ into the product $\taj$. We get 
\begin{multline*}
T (\tbkl) = \avsumj \ajpsi \\
\oneoverroottwo \Bigg( \pm \Big( \signjk \tajmodk \otimes I \tketpsi \Big)^\dagger  \\
+ \Big( \signjl \tajmodl \otimes I \tketpsi \Big)^\dagger \Bigg)
\end{multline*}
We separate into two sums and change the index of summation in each and we get
\begin{multline*}
T (\tbkl) = \avsumj \oneoverroottwo \Bigg( \pm \ajpsimodk  \\
+ \ajpsimodl \Bigg) \tildepsiajinv
\end{multline*}

From here, the identity 
\begin{multline*}
(\bkl) T - T (\tbkl) = \avsumj \Bigg( \ajbklpsi  \\
 - \oneoverroottwo \Big( \pm \ajpsimodk  \\
 + \ajpsimodl \Big) \Bigg) \tildepsiajinv
\end{multline*}
follows.

\subsection{The first error bound}\label{subsec:the_first_error_bound}

Here, we aim to show that for all $i \in \{1, \dots n\}$, for all $\jinzeroonetothen$, 
\begin{multline}\label{eq:first_error_bound}
\Big\| A_i \ajpsi - \ajpsimodi \Big\|   \\
\leq (6 + 4 \sqrt{2}) n^2 \sqrt{\epsilon} < 12 n^2 \sqrt{\epsilon}
\end{multline}

The situation is the following: we would like to insert $A_i$ into the product $\aj$ as if the \aobservables were anti-commuting. However, we don't know that \aobservables are anti-commuting; all we know about the \aobservables is that they are part of an $\epsilon$-optimal \chshn strategy. 

The first step is to recognize that even though \aobservables may not be anti-commuting as operators, they nearly anti-commute in their action on the strategy state $\ketpsi$. We prove the following:

\begin{lemma} \label{lemma:approximate_anti_commutation_on_psi}
Let \abpsi be an $\epsilon$-optimal \chshn strategy. Then, 
\[ \sumij \left\| \frac{A_i A_j + A_j A_i}{2} \otimes I \ketpsi \right\|^2 \leq (1+\sqrt{2})^2 n(n-1) \epsilon \]
\end{lemma}

\begin{proof}
We recognize that the operators
\[ \frac{A_i + A_j}{\sqrt{2}} \otimes I + I \otimes B_{ij}\]
and
\[ \frac{A_i - A_j}{\sqrt{2}} \otimes I +I \otimes B_{ji}\] 
each have operator norm at most $(1+\sqrt{2})$, by the triangle inequality.  

Next, we see that 
\begin{multline*}
\left\| \frac{A_i A_j + A_j A_i}{2} \otimes I \ketpsi \right\|  \\
= \left\| \left(\frac{A_i + A_j}{\sqrt{2}} \otimes I + I \otimes B_{ij} \right) \left( \frac{A_i + A_j}{\sqrt{2}} \otimes I - I \otimes B_{ij} \right) \ketpsi \right\| \\
\leq (1+\sqrt{2}) \left\| \left( \frac{A_i + A_j}{\sqrt{2}} \otimes I - I \otimes B_{ij} \right) \ketpsi \right\|
\end{multline*}
and similarly, 
\begin{multline*}
\left\| \frac{A_i A_j + A_j A_i}{2} \otimes I \ketpsi \right\|  \\
= \left\| \left(\frac{A_i - A_j}{\sqrt{2}} \otimes I + I \otimes B_{ji} \right) \left( \frac{A_i - A_j}{\sqrt{2}} \otimes I - I \otimes B_{ji} \right) \ketpsi \right\| \\
\leq (1+\sqrt{2}) \left\| \left( \frac{A_i - A_j}{\sqrt{2}} \otimes I - I \otimes B_{ji} \right) \ketpsi \right\|
\end{multline*}

Now we use the relation 
\begin{multline*} 
\sum_{1 \leq i < j \leq n} \Bigg( \left\| \Big( \frac{A_i + A_j}{\sqrt{2}} \otimes I - I \otimes B_{ij} \Big) \ketpsi \right\|^2  
\\ + \left\| \Big( \frac{A_i - A_j}{\sqrt{2}} \otimes I - I \otimes B_{ji} \Big) \ketpsi \right\|^2 \Bigg)  \leq 2 n (n-1) \epsilon 
\end{multline*}
from Theorem \ref{thm:chshn_equality_conditions}. We get 
\begin{multline*}
\sumij 2 \left\| \frac{A_i A_j + A_j A_i}{2} \otimes I \ketpsi \right\|^2 \\
\leq (1+\sqrt{2})^2 \sum_{1 \leq i < j \leq n} \Bigg( \left\| \frac{A_i + A_j}{\sqrt{2}} \otimes I \ketpsi - I \otimes B_{ij} \ketpsi \right\|^2 \\
+ \left\| \frac{A_i - A_j}{\sqrt{2}} \otimes I \ketpsi - I \otimes B_{ji} \ketpsi \right\|^2 \Bigg) \\ \leq  (1+\sqrt{2})^2 2 n (n-1) \epsilon 
\end{multline*}
The lemma is proved. 
\end{proof}

Now we know that \aobservables almost anti-commute in their action on the strategy state $\ketpsi$. This is a step forward, but still not enough for proving the bound \eqref{eq:first_error_bound}. To see why, consider a product like $A_i A_1 A_2 \otimes I \ketpsi$. We want to switch the order of $A_i$ and $A_1$. We know that $A_i$ and $A_1$ nearly anti-commute in their action on $\ketpsi$, but we don't yet know that they nearly anti-commute in their action on $A_2 \otimes I \ketpsi$. 

Fortunately, this difficulty can be circumvented: we know from Theorem \ref{thm:chshn_equality_conditions} that, for example, $A_2 \otimes I \ketpsi \approx I \otimes \oneoverroottwo (B_{12} - B_{21}) \ketpsi$. This helps, because
\begin{multline*}
(A_i A_1 \otimes I) (A_2 \otimes I) \ketpsi \approx  (A_i A_1 \otimes I) (I \otimes \oneoverroottwo (B_{12} - B_{21})) \ketpsi \\
= (I \otimes \oneoverroottwo (B_{12} - B_{21})) (A_i A_1 \otimes I) \ketpsi
\end{multline*}
and now we can switch the order of $A_i$ and $A_1$ in their action on $\ketpsi$. 

The preceding discussion shows that we can use the anti-commutation on $\ketpsi$ (Lemma \ref{lemma:approximate_anti_commutation_on_psi}) to switch the order of a product of the $A_i$'s acting on $\ketpsi$, as long as we can "get some of the $A_i$'s out of the way", by replacing their action with the action of an operator on the $B$ side. 

For reason of keeping the errors of approximation under control, we want the operators on the $B$ side that we use to have operator norm 1. The operators $\oneoverroottwo (B_{ij} \pm B_{ji})$ do not necessarily have operator norm 1, but fortunately this difficulty can also be circumvented. 

The discussion in the previous paragraphs motivates us to prove the following lemma:

\begin{lemma} \label{lemma:best_AB_switch}
Fix $k$. Then, there exists an $l$ such that
\[ \left\| A_k \otimes I \ketpsi - I \otimes \frac{\pm B_{kl} + B_{lk}}{|\pm B_{kl} + B_{lk}|} \ketpsi \right\| \leq (2 \sqrt{2} + 2) \sqrt{n} \sqrt{\epsilon} \] 
where $+B_{kl}$ is taken if $l>k$ and $-B_{kl}$ is taken if $l<k$. The notation 
\[ \frac{\pm B_{kl} + B_{lk}}{|\pm B_{kl} + B_{lk}|} \]
means that we take all eigenvalues of the operator $\pm B_{kl} + B_{lk}$ and normalize the positive ones to $1$, the negative ones to $-1$, and, by convention, the eigenvalue $0$ gets normalized to $1$. 
\end{lemma}

\begin{proof}
The proof proceeds in two steps: first, we approximate $A_k \otimes I \ketpsi$ by $ I \otimes \frac{\pm B_{kl} + B_{lk}}{\sqrt{2}} \ketpsi$ and then we approximate $ I \otimes \frac{\pm B_{kl} + B_{lk}}{\sqrt{2}} \ketpsi$ by $I \otimes \frac{\pm B_{kl} + B_{lk}}{|\pm B_{kl} + B_{lk}|} \ketpsi$. 

We prove the first step. We take the relation 
\begin{multline*}
\sum_{1 \leq i < j \leq n} \Bigg( \left\| A_i \otimes I \ketpsi - I \otimes \frac{B_{ij} + B_{ji}}{\sqrt{2}} \ketpsi \right\|^2 \\
+ \left\| A_j \otimes I \ketpsi - I \otimes \frac{B_{ij} - B_{ji}}{\sqrt{2}} \ketpsi \right\|^2 \Bigg) \leq 2 n (n-1) \epsilon 
\end{multline*}
from Theorem \ref{thm:chshn_equality_conditions}. We focus only on those terms of the sum that contain $A_k$ and we get
\begin{multline*}
\sum_{j=k+1}^n \left\|A_k \otimes I \ketpsi - I \otimes \frac{B_{kj}+B_{jk}}{\sqrt{2}} \ketpsi \right\|^2 \\
+ \sum_{j=1}^{k-1} \left\|A_k \otimes I \ketpsi - I \otimes \frac{- B_{kj}+B_{jk}}{\sqrt{2}} \ketpsi \right\|^2 \leq 2 n (n-1) \epsilon
\end{multline*}
Pick the smallest of these $(n-1)$ terms. It satisfies
\begin{equation} \label{eq:bound_for_first_step_of_AB_switch}
\left\|A_k \otimes I \ketpsi - I \otimes \frac{\pm B_{kl}+B_{lk}}{\sqrt{2}} \ketpsi \right\|^2 \leq 2 n \epsilon
\end{equation}
This is how we approximate $A_k \otimes I \ketpsi$ by $ I \otimes \frac{\pm B_{kl} + B_{lk}}{\sqrt{2}} \ketpsi$. 

Next we focus on the second step. By Lemma \ref{lemma:the_business_with_normalizing} which we will prove below, 
\begin{equation} \label{eq:bound_for_second_step_of_AB_switch}
\left\| I \otimes \frac{\pm B_{kj}+B_{jk}}{\sqrt{2}} \ketpsi - I \otimes \frac{\pm B_{kl} + B_{lk}}{|\pm B_{kl} + B_{lk}|} \ketpsi \right\| \leq \left\| I \otimes \frac{B_{kl} B_{lk} + B_{lk} B_{kl}}{2} \ketpsi \right\|
\end{equation} 
so it suffices to give a bound on \[\left\| I \otimes \frac{B_{kl} B_{lk} + B_{lk} B_{kl}}{2} \ketpsi \right\|\] 

For the bound on  $\left\| I \otimes \frac{B_{kl} B_{lk} + B_{lk} B_{kl}}{2} \ketpsi \right\|$, we observe that the operator 
\[ A_k \otimes I + I \otimes \frac{\pm B_{kl}+B_{lk}}{\sqrt{2}}  \]
has operator norm at most $(1+\sqrt{2})$, and so 
\begin{multline*}
\left\| I \otimes \frac{B_{kl} B_{lk} + B_{lk} B_{kl}}{2} \ketpsi \right\| \\
= \left\| \left( A_k \otimes I + I \otimes \frac{\pm B_{kl}+B_{lk}}{\sqrt{2}} \right) \left( A_k \otimes I - I \otimes \frac{\pm B_{kl}+B_{lk}}{\sqrt{2}} \right) \ketpsi \right\|  \\
\leq (1+ \sqrt{2}) \left\|A_k \otimes I \ketpsi - I \otimes \frac{\pm B_{kl}+B_{lk}}{\sqrt{2}} \ketpsi \right\| \leq (1+ \sqrt{2}) \sqrt{2 n \epsilon}
\end{multline*}
Combining this with inequalities \eqref{eq:bound_for_first_step_of_AB_switch} and \eqref{eq:bound_for_second_step_of_AB_switch}, we get that 
\[ \left\| A_k \otimes I \ketpsi - I \otimes \frac{\pm B_{kl} + B_{lk}}{|\pm B_{kl} + B_{lk}|} \ketpsi \right\| \leq (2 \sqrt{2} + 2) \sqrt{n} \sqrt{\epsilon} \] 
as needed. The lemma is proved. 
\end{proof}

Next, we prove the missing link in the proof of Lemma \ref{lemma:best_AB_switch}, which has to do with operators of the form $\frac{R+S}{\sqrt{2}}$ and $\frac{R+S}{|R+S|}$ when $R$, $S$ are $\pm 1$ observables.  

\begin{lemma}\label{lemma:the_business_with_normalizing}
Let $R, S$ be two $\pm 1$ observables on $\C^d$. Then, 
\begin{enumerate}
\item The following operator identity holds:
\begin{multline*} \left( \frac{R+S}{\sqrt{2}} - \frac{R+S}{|R+S|} \right)^2  \\
= \left( \frac{RS+SR}{2} \right) \left(  2 I + \frac{RS+SR}{2} + 2 \sqrt{I + \frac{RS+SR}{2}}  \right)^{-1} \left( \frac{RS+SR}{2} \right) 
\end{multline*}
\item The operator
\[ \left( \frac{RS+SR}{2} \right)^2 - \left( \frac{R+S}{\sqrt{2}} - \frac{R+S}{|R+S|} \right)^2 \]
is positive semi-definite. 
\item For any vector $\ket{v}$, 
\[ \left\| \frac{R+S}{\sqrt{2}} \ket{v} - \frac{R+S}{|R+S|} \ket{v} \right\| \leq \left\| \frac{RS+SR}{2} \ket{v} \right\| \]
\end{enumerate}
\end{lemma}

\begin{proof}
We first prove the operator identity. 

We break up $\C^d$ into eigenspaces for the self-adjoint operator $R+S$. Since $RS+SR = (R+S)^2 - 2I$, these are also eigenspaces for the operator $RS+SR$, and so also eigenspaces for the operator
\[ \left( \frac{RS+SR}{2} \right) \left(  2 I + \frac{RS+SR}{2} + 2 \sqrt{I + \frac{RS+SR}{2}}  \right)^{-1} \left( \frac{RS+SR}{2} \right) \]
We will prove that the operator identity holds on each of the aforementioned eigenspaces. 

Consider an eigenspace where $R+S$ has eigenvalue $\lambda$. 

On this eigenspace, the operator 
\[\left( \frac{R+S}{\sqrt{2}} - \frac{R+S}{|R+S|} \right)^2\]
has eigenvalue $\left( \frac{(sign \lambda) \lambda}{\sqrt{2}} - 1 \right)^2$; this holds in all the three cases $\lambda>0, \: \lambda < 0, \: \lambda=0$. 

The eigenvalue of $\left( \frac{RS+SR}{2} \right)$ on this eigenspace is $\frac{\lambda^2 - 2}{2}$.

The eigenvalue of 
\[ \left( \frac{RS+SR}{2} \right) \left(  2 I + \frac{RS+SR}{2} + 2 \sqrt{I + \frac{RS+SR}{2}}  \right)^{-1} \left( \frac{RS+SR}{2} \right) \]
is therfore
\[ \left( \frac{\lambda^2 - 2}{2} \right)^2 \frac{1}{2+\frac{\lambda^2 - 2}{2} + 2 \sqrt{1+\frac{\lambda^2 - 2}{2}}} \]

Next, we observe that 
\begin{multline*}
2+\frac{\lambda^2 - 2}{2} + 2 \sqrt{1+\frac{\lambda^2 - 2}{2}} = 1+ \frac{\lambda^2}{2} + 2 \sqrt{\frac{\lambda^2}{2}} = \left( \frac{(sign \lambda) \lambda}{\sqrt{2}} + 1 \right)^2
\end{multline*}
and that 
\[ \left( \frac{\lambda^2 - 2}{2} \right)^2 = \left( \frac{(sign \lambda) \lambda}{\sqrt{2}} - 1 \right)^2 \left( \frac{(sign \lambda) \lambda}{\sqrt{2}} + 1 \right)^2 \]
and therefore,
\[ \left( \frac{\lambda^2 - 2}{2} \right)^2 \frac{1}{2+\frac{\lambda^2 - 2}{2} + 2 \sqrt{1+\frac{\lambda^2 - 2}{2}}} = \left( \frac{(sign \lambda) \lambda}{\sqrt{2}} - 1 \right)^2 \]

Next we use the above to conclude that the operators
\[\left( \frac{R+S}{\sqrt{2}} - \frac{R+S}{|R+S|} \right)^2\]
and 
\[ \left( \frac{RS+SR}{2} \right) \left(  2 I + \frac{RS+SR}{2} + 2 \sqrt{I + \frac{RS+SR}{2}}  \right)^{-1} \left( \frac{RS+SR}{2} \right) \]
have the same eigenvalue on this eigenspace. 

This argument holds for any eigenspace, and so the operator identity 
\begin{multline*} \left( \frac{R+S}{\sqrt{2}} - \frac{R+S}{|R+S|} \right)^2  \\
= \left( \frac{RS+SR}{2} \right) \left(  2 I + \frac{RS+SR}{2} + 2 \sqrt{I + \frac{RS+SR}{2}}  \right)^{-1} \left( \frac{RS+SR}{2} \right) 
\end{multline*}
holds. 

Next we prove the second part. We can see from the argument above that the operator 
\[ \left(  2 I + \frac{RS+SR}{2} + 2 \sqrt{I + \frac{RS+SR}{2}}  \right)^{-1} \] 
has eigenvalues of the form
\[ \frac{1}{ \left( \frac{(sign \lambda) \lambda}{\sqrt{2}} + 1 \right)^2} \]
and they are all in $(0,1]$. Therefore, 
\begin{multline*} \left( \frac{R+S}{\sqrt{2}} - \frac{R+S}{|R+S|} \right)^2  \\
= \left( \frac{RS+SR}{2} \right) \left(  2 I + \frac{RS+SR}{2} + 2 \sqrt{I + \frac{RS+SR}{2}}  \right)^{-1} \left( \frac{RS+SR}{2} \right)  \\
\preceq \left( \frac{RS+SR}{2} \right)^2
\end{multline*}

Finally, we observe that the third part follows directly from the second. The lemma is proved. 
\end{proof}

Recall that the goal of this section is to prove
\begin{multline*}
\Big\| A_i \ajpsi - \ajpsimodi \Big\|  \\
\leq (6 + 4 \sqrt{2}) n^2 \sqrt{\epsilon} < 12 n^2 \sqrt{\epsilon}
\end{multline*}
and the overall strategy is to insert $A_i$ into the product $\aj$ as if the \aobservables were anti-commuting.

The results of the lemmas above have prepared the tools necessary for this goal. Lemma \ref{lemma:approximate_anti_commutation_on_psi} tells us that 
\[ A_k A_l \otimes I \ketpsi \approx - A_l A_k \otimes \ketpsi \]
with the error of approximation being at most $(2 \sqrt{2} + 2) n \sqrt{\epsilon}$. We call this apporoximation step an anticommutation switch.  Lemma \ref{lemma:best_AB_switch} tells us that 
\[ A_k \otimes I \ketpsi \approx I \otimes \frac{\pm B_{kl} + B_{lk}}{|\pm B_{kl} + B_{lk}|} \ketpsi \]
where $\frac{\pm B_{kl} + B_{lk}}{|\pm B_{kl} + B_{lk}|}$ is a suitable $\pm 1$ observable acting on the $B$ side, and the error of approximation is at most $(2 \sqrt{2} + 2) \sqrt{n} \sqrt{\epsilon}$. We call this approximation step an AB-switch. 

The idea is to concatenate a number of these approximation steps to get the bound \eqref{eq:first_error_bound}. We present a procedure that goes from \[A_i \ajpsi\] to \[\ajpsimodi\] using at most $n$ anti-commutator switches and $2n$ AB-switches. The procedure is the following:
\begin{enumerate}
\item Start with $A_i \ajpsi$.
\item Switch all elements of the product $\aj$ to the $B$ side using the AB-switches.
\item Repeat
\begin{enumerate}
\item Switch the last observable on the $B$ side back to the $A$ side
\item Anti-commute $A_i$ and the newly switched observable
\end{enumerate}
until $A_i$ comes to its proper position. 
\item Switch the observables still remaining on the $B$ side back to the $A$ side. 
\end{enumerate}

The total approximation error of this procedure is at most
\[ n (2 \sqrt{2} + 2) n \sqrt{\epsilon} + (2n) (2 \sqrt{2} + 2) \sqrt{n} \sqrt{\epsilon} \leq (6 + 4 \sqrt{2}) n^2 \sqrt{\epsilon} < 12 n^2 \sqrt{\epsilon} \]
The bound 
\begin{multline*}
\Big\| A_i \ajpsi - \ajpsimodi \Big\|  \\
\leq (6 + 4 \sqrt{2}) n^2 \sqrt{\epsilon} < 12 n^2 \sqrt{\epsilon}
\end{multline*}
is proved.

\subsection{The second error bound} \label{subsec:the_second_error_bound}
The goal of this subsection is to prove that 
\begin{multline*}
\Bigg\| \ajbklpsi 
 - \oneoverroottwo \Big( \pm \ajpsimodk  \\
 + \ajpsimodl \Big) \Bigg\|  \\
 \leq \left( \frac{17}{2} + 6 \sqrt{2} \right) n^2 \sqrt{\epsilon} < 17 n^2 \sqrt{\epsilon}
\end{multline*}

The argument is similar to the previous subsection. 

By the triangle inequality, we have 
\begin{multline*}
\Bigg\| \ajbklpsi  - \oneoverroottwo \Big( \pm \ajpsimodk  \\
+ \ajpsimodl \Big) \Bigg\|  \\
\leq \left\| \ajbklpsi - \aj \frac{\pm A_k +A_l}{\sqrt{2}} \otimes I \ketpsi \right\|  \\
+ \oneoverroottwo \left\| \aj A_k \otimes I \ketpsi - \ajpsimodk \right\|  \\
+ \oneoverroottwo \left\| \aj A_l \otimes I \ketpsi - \ajpsimodl \right\|
\end{multline*}

For the first term we have the following:
\begin{multline*}
\left\| \ajbklpsi - \aj \frac{\pm A_k +A_l}{\sqrt{2}} \otimes I \ketpsi \right\|  \\
= \left\| I \otimes B_{kl} \ketpsi - \frac{\pm A_k +A_l}{\sqrt{2}} \otimes I \ketpsi \right\| \leq \sqrt{2 n (n-1) \epsilon}
\end{multline*} 
where we have used the inequalities in Theorem \ref{thm:chshn_equality_conditions}. 

For the second term, we claim that 
\begin{multline*} \left\| \aj A_k \otimes I \ketpsi - \ajpsimodk \right\| \\
\leq (6 + 4 \sqrt{2}) n^2 \sqrt{\epsilon} 
\end{multline*}
The argument is similar to the argument in the previous section: we present a procedure that goes from \[\aj A_k \otimes I \ketpsi\] to \[\ajpsimodk\] using at most $n$ anti-commutator switches and $2n$ AB-switches. The procedure is the following:
\begin{enumerate}
\item Start with $\aj A_k \otimes I \ketpsi$.
\item Repeat
\begin{enumerate}
\item Anti-commute $A_k$ and the next to last observable on the $A$ side
\item Move the newly switched observable to the $B$ side
\end{enumerate}
until $A_k$ comes to its proper position. 
\item Switch the observables still remaining on the $B$ side back to the $A$ side. 
\end{enumerate}

The third term is analyzed in the same manner and we get 
\begin{multline*} \left\| \aj A_l \otimes I \ketpsi - \ajpsimodl \right\| \\
\leq (6 + 4 \sqrt{2}) n^2 \sqrt{\epsilon} 
\end{multline*}

Combining all the preceding bounds, we get 
\begin{multline*}
\Bigg\| \ajbklpsi - \oneoverroottwo \Big( \pm \ajpsimodk  \\
+ \ajpsimodl \Big) \Bigg\|  \\
\leq \sqrt{2 n (n-1) \epsilon} + \oneoverroottwo (6 + 4 \sqrt{2}) n^2 \sqrt{\epsilon} 
+ \oneoverroottwo (6 + 4 \sqrt{2}) n^2 \sqrt{\epsilon} \\
\leq \left( \frac{17}{2} + 6 \sqrt{2} \right) n^2 \sqrt{\epsilon} < 17 n^2 \sqrt{\epsilon}
\end{multline*}
as needed. 

\subsection{Putting everything together}
The aim of this subsection is to put all the previous steps together and show that 
\begin{align*}
\forall i \quad \|(A_i \otimes I) T - T (\tilde{A}_i \otimes I) \|_F &< 12 n^2 \sqrt{\epsilon} \|T\|_F \\
\forall j \neq k \quad  \| (I \otimes B_{jk}) T - T (\tbjk) \|_F &< 17 n^2 \sqrt{\epsilon} \|T\|_F
\end{align*}
thereby completing the proof of Theorem \ref{thm:chshn_near_optimal_strategies}.

We start with the first inequality. We know from subsection \ref{subsec:the_expression_for_tai} that 
\begin{multline*}
(\ai) T - T (\tai) = \avsumj \Big( A_i \ajpsi \\
- \ajpsimodi \Big) \tildepsiajinv
\end{multline*}

We know from subsection \ref{subsec:orthonormal_vectors} that the vectors 
\[ \left\{ \tildeajpsi \: : \: \jinzeroonetothen \right\} \]
are orthonormal. 

We also know from subsection \ref{subsec:the_first_error_bound} that for all $i$, for all $\jinzeroonetothen$
\begin{multline*}
\Big\| A_i \ajpsi - \ajpsimodi \Big\|   \\
\leq (6 + 4 \sqrt{2}) n^2 \sqrt{\epsilon} < 12 n^2 \sqrt{\epsilon}
\end{multline*}

We combine these facts using Lemma \ref{lemma:calculating_frobenius_norm} and we get that for all $i$, 
\[ \|(A_i \otimes I) T - T (\tilde{A}_i \otimes I) \|_F < 12 n^2 \sqrt{\epsilon} = 12 n^2 \sqrt{\epsilon} \|T\|_F \]
where in the last step we have used the fact that $T$ was chosen so that $\|T\|_F = 1$ (subsection \ref{subsec:frobenius_norm_of_t}). 

In a similar manner, we take the results of subsections \ref{subsec:orthonormal_vectors}, \ref{subsec:frobenius_norm_of_t}, \ref{subsec:the_expression_for_tbkl}, and \ref{subsec:the_second_error_bound} and apply Lemma \ref{lemma:calculating_frobenius_norm} and get that for all $j \neq k \in \{1, \dots n\}$ 
\[ \| (I \otimes B_{jk}) T - T (\tbjk) \|_F < 17 n^2 \sqrt{\epsilon} \|T\|_F \]
The proof of Theorem \ref{thm:chshn_near_optimal_strategies} is complete. 

\section{Conclusion and open problems}\label{ch:open_problems}

In this paper, we first derived a general result about non-local XOR games: for every non-local XOR game, there exists a set of relations such that a quantum strategy is optimal if and only if it satisfies the relations and a quantum strategy is nearly-optimal if and only if it approximately satisfies the relations. Then, we focused on the \chshn XOR games, and derived the structure of their optimal and nearly-optimal quantum strategies. 

One possible direction for future work is whether structure results like the one for \chshn near-optimal quantum strategies can be proved for other non-local games. The \chshn games have a very regular structure, and the arguments above make heavy use of this structure; however, it may be possible to construct an argument of this form, or another form altogether, for other XOR games with less regular structure. 

Another possible direction for future work is whether the \chshn games can be used in the context of quantum information processing with untrusted black-box devices. The CHSH game, the first member of the \chshn family, has already been used in protocols for doing information processing with untrusted devices. Whether all the \chshn games can be used, and which of the \chshn games gives protocols with the best parameters, are two questions that are still open.

\section*{Acknowledgements}

The results of this paper first appear in my PhD thesis submitted to the Department of Mathematics at Massachusetts Institute of Technology. The material is  used here with permission from MIT. 

I would like to thank my thesis advisor Prof. Peter Shor for his unconditional support through the years. Prof. Shor gave me the freedom I needed to explore, and to find my own way. He was also generous with his time, and patiently listened to my mathematical arguments and ideas. 
 
I would like to thank Thomas Vidick for bringing to my attention the problem of self-testing and entanglement rigidity. Thomas has always been friendly, enthusiastic, and open to discussion. The conversations with him have been a source of many great ideas.

\end{document}